\documentclass[11pt]{article}
\usepackage[utf8]{inputenc}
\usepackage{amsmath}
\usepackage{amsfonts}
\usepackage{amssymb}
\usepackage{amsthm}
\usepackage[unicode=true,pdfusetitle,
 bookmarks=true,bookmarksnumbered=false,bookmarksopen=false,
 breaklinks=false,pdfborder={0 0 1},backref=page,colorlinks=true]
 {hyperref}
\usepackage{color}
\usepackage{graphicx}
\usepackage{subcaption}
\usepackage[dvipsnames,svgnames,x11names,hyperref]{xcolor}
\usepackage[round]{natbib}
\usepackage[margin=1in]{geometry}
\newtheorem{theorem}{Theorem}
\newtheorem{proposition}{Proposition}
\newtheorem{lemma}{Lemma}
\usepackage{dsfont}
\hypersetup{linkcolor=RoyalBlue,citecolor=RoyalBlue,urlcolor=RoyalBlue}
\usepackage{tikz}
\usepackage{caption}
\usepackage{float}
\usetikzlibrary{arrows.meta}
\usetikzlibrary{calc}

\usepackage{comment}
\def\rmd{\mathrm{d}}
\usepackage{times}
\def\T{{ \mathrm{\scriptscriptstyle T} }}
\def\rmd{\mathrm{d}}

\begin{document}
\title{Simulating Diffusion Bridges with Score Matching}
\author{Jeremy Heng\thanks{ESSEC Business School, Singapore 139408, Singapore} \footnote{Corresponding author: heng@essec.edu}
\and Valentin De Bortoli\thanks{Center for Science of Data, ENS Ulm, Paris 75005, France} 
\and Arnaud Doucet\thanks{Department of Statistics, University of Oxford, Oxford OX1 3LB, U.K.} 
\and James Thornton\footnotemark[4]}
\maketitle

\begin{abstract}
We consider the problem of simulating diffusion bridges, which are diffusion processes that are conditioned to initialize and terminate at two given states. The simulation of diffusion bridges has applications in diverse scientific fields and plays a crucial role in the statistical inference of discretely-observed diffusions. This is known to be a challenging problem that has received much attention in the last two decades. This article contributes to this rich body of literature by presenting a new avenue to obtain diffusion bridge approximations. Our approach is based on a backward time representation of a diffusion bridge, which may be simulated if one can time-reverse the unconditioned diffusion. We introduce a variational formulation to learn this time-reversal with function approximation and rely on a score matching method to circumvent intractability. Another iteration of our proposed methodology approximates the Doob's $h$-transform defining the forward time representation of a diffusion bridge. We discuss algorithmic considerations and extensions, and present numerical results on an Ornstein--Uhlenbeck process, a model from financial econometrics for interest rates, and a model from genetics for cell differentiation and development to illustrate the effectiveness of our approach. 
\end{abstract}
\textbf{\small{}Keywords}{\small{}: Diffusion; Diffusion bridge; Score matching; Stochastic differential equation; Time-reversal. }{\small \par}

\section{Introduction}
Diffusion processes have been used extensively in mathematical and natural sciences. 
A diffusion process $X=(X_t)_{t\in[0,T]}$ in $\mathbb{R}^d$ is defined by the stochastic differential equation 
\begin{align}\label{eqn:SDE}
    \rmd X_t = f(t,X_t)\rmd t + \sigma(t,X_t)\rmd W_t,
\end{align}
where $f:[0,T]\times\mathbb{R}^d\rightarrow\mathbb{R}^d$ is a drift function, $\sigma:[0,T]\times\mathbb{R}^d\rightarrow\mathbb{R}^{d\times d}$ is a diffusion coefficient, and $W=(W_t)_{t\in[0,T]}$ is a $d$-dimensional Brownian motion. 
We suppose $f$ and $\sigma$ are sufficiently regular to induce a unique weak solution and $\Sigma(t,x_t)=(\sigma\sigma^\T)(t,x_t)$ is uniformly positive definite for all $(t,x_t)\in[0,T]\times\mathbb{R}$. For any $0\leq s<t\leq T$, we denote the transition density of \eqref{eqn:SDE} with respect to the Lebesgue measure on $\mathbb{R}^d$ as $p(t,x_t\mid s,x_s)$ and assume that it is positive for ease of exposition. While the numerical simulation of $X$ can be routinely handled by time-discretization schemes \citep{kloeden1992}, the task of simulating $X$ initialized at $X_0=x_0$ and conditioned to terminate at $X_T=x_T$ is a challenging problem that has received much attention in the last two decades. 

Simulating the conditioned process $X^{\star}=(X_t^{\star})_{t\in[0,T]}$, commonly referred to as a diffusion bridge, has applications in diverse fields such as computational chemistry \citep{bolhuis2002transition,wang2020exact}, financial econometrics \citep{elerian2001likelihood,durham2002numerical}, genetics \citep{wang2011quantifying}, and shape analysis \citep{arnaudon2020diffusion}. When performing statistical inference for parameters of $f$ and $\sigma$ in the case where $X$ is observed at discrete time points, diffusion bridge simulation is a crucial tool that allows one to impute missing paths between observations within an expectation-maximization algorithm or a Gibbs sampler \citep{pedersen1995consistency,roberts2001inference,eraker2001mcmc,beskos2006exact,golightly2008bayesian,van2017bayesian}. 

By Doob's $h$-transform \citep[p. 83]{rogers2000diffusions}, it is well-known that $X^{\star}$ satisfies
\begin{align}\label{eqn:doob_SDE}
    \rmd X_t^{\star} = \{f(t,X_t^{\star}) + \Sigma(t,X_t^{\star})\nabla \log h(t,X_t^{\star})\}\rmd t + \sigma(t,X_t^{\star})\rmd W_t, \quad X_0^{\star}=x_0,
\end{align}
where $h(t,x_t) = p(T,x_T\mid t,x_t)$ and $\nabla$ denotes the gradient operator. The term $\Sigma(t,x_t)\nabla \log h(t,x_t)$ forces the conditioned process towards the terminal condition $X_T^{\star}=x_T$. As the transition density and hence its logarithmic gradient is intractable for most diffusions, exploiting this result to simulate diffusion bridges is highly non-trivial. To this end, one can characterize $h$ as the solution of the backward Kolmogorov equation 
\begin{align}\label{eqn:backward_kolmogorov}
	\partial_t h(t,x_t) + (\mathcal{L}h)(t,x_t) = 0,
\end{align}
with terminal condition at time $T$ given by the Dirac measure at $x_T$, where $\mathcal{L}$ denotes the generator of $X$ \citep{stroock1997multidimensional}. Equation \eqref{eqn:backward_kolmogorov} reveals that $h$ propagates information about the terminal constraint backwards in time. 
However, numerical resolution of this partial differential equation is particularly challenging due to the singularity at time $T$, and computationally demanding when the dimension $d$ is large \citep{wang2020exact}. Furthermore, one must run a solver for every pair of conditioned states $(x_0,x_T)$ considered. 

A common approach to address these difficulties is to simulate a proposal bridge process $X^{\circ}=(X_t^{\circ})_{t\in[0,T]}$, satisfying $\rmd X_t^{\circ} = f^{\circ}(t,X_t^{\circ}) \rmd t + \sigma(t,X_t^{\circ})\rmd W_t$ with $X_0^{\circ}=x_0$. 
One constructs $f^{\circ}:[0,T]\times\mathbb{R}^d\rightarrow\mathbb{R}^d$ using a tractable approximation of \eqref{eqn:doob_SDE}, and corrects for the discrepancy using importance sampling or an independent Metropolis--Hastings algorithm \citep{papaspiliopoulos2012importance,elerian2001likelihood}. 
The simple choice $f^{\circ}(t,x_t)=f(t,x_t)$ typically performs poorly as it does not take the constraint $X_T=x_T$ into account \citep{pedersen1995consistency}. 
The drift of a Brownian bridge $f^{\circ}(t,x_t) = (x_T-x_t)/(T-t)$ has been considered in several works \citep{durham2002numerical,delyon2006simulation,stramer2007simulated,papaspiliopoulos2013data}, and improved by \citet{whitaker2017improved} using an innovative decomposition of the process into deterministic and stochastic parts.
\citet{clark1990simulation} followed by \citet{delyon2006simulation} studied the choice 
$f^{\circ}(t,x_t) = f(t,x_t) + (x_T-x_t)/(T-t)$ that incorporates the dynamics of the original process $X$. 
To introduce more flexibility and better mimic the structure of \eqref{eqn:doob_SDE}, \citet{schauer2017guided} proposed setting $f^{\circ}(t,x_t) = f(t,x_t) + \Sigma(t,x_t)\nabla\log \tilde{h}(t,x_t)$, where $\tilde{h}(t,x_t) = \tilde{p}(T,x_T \mid t,x_t)$ is an analytically tractable transition density of an auxiliary process. For tractability, the latter is typically chosen from the class of linear processes and can be optimized to get the best approximation within this class \citep{van2017bayesian}.
Other Markov chain Monte Carlo approaches include Gibbs sampling \citep{eraker2001mcmc}, Langevin-type stochastic partial differential equations \citep{stuart2004conditional,beskos2008mcmc}, piecewise deterministic Monte Carlo \citep{bierkens2021piecewise}, and manifold Hamiltonian Monte Carlo methods \citep{graham2022manifold}.

The exact simulation algorithms developed in \citet{beskos2005exact} and \citet{beskos2006retrospective} can be employed to sample diffusion bridges without any time-discretization error. However, these elegant methods are limited to the class of diffusion processes that can be transformed to have unit diffusion coefficient. 
\citet{bladt2014simple} and \citet{bladt2016simulation} devised an ingenious methodology to simulate diffusion bridges based on coupling and time-reversal of diffusions. Their proposed method is applicable to the class of ergodic diffusions with an invariant density that is either explicitly known or numerically approximated.
Closely related approaches include sequential Monte Carlo algorithms that resample using backward information filter approximations \citep{guarniero2017iterated}, information from backward pilot paths \citep{lin2010generating}, or guided weight functions \citep{del2015sequential}. The main idea underlying these works is the representation of the diffusion bridge in \eqref{eqn:doob_SDE} and \eqref{eqn:backward_kolmogorov}.

\section{Diffusion bridges}

\subsection{Time-reversed bridge process}\label{sec:time_reversed_bridge}
It can be shown that the time-reversed bridge process $Z^{\star}=(Z_t^{\star})_{t\in[0,T]}=(X_{T-t}^{\star})_{t\in[0,T]}$ satisfies
\begin{align}\label{eqn:reversal_bridge}
    \rmd Z_t^{\star} = b(t,Z_t^{\star})\rmd t + \sigma(T-t,Z_t^{\star}) \rmd B_t,\quad Z_0^{\star} = x_T,
\end{align}
with drift function $b(t,z_t) = -f(T-t,z_t) + \Sigma(T-t,z_t)s(T-t,z_t) + \nabla \cdot\Sigma(T-t,z_t)$, another standard Brownian motion $B=(B_t)_{t\in[0,T]}$, $s(t,x_t) = s^{\star}(t,x_t) - \nabla \log h(t,x_t)$, and $\nabla\cdot\Sigma(t,x)=(\sum_{j=1}^d\partial_{x_j}\Sigma^{1,j}(t,x), \ldots,\sum_{j=1}^d\partial_{x_j}\Sigma^{d,j}(t,x))$ is the divergence of $\Sigma$. 
Here $s^{\star}(t,x_t)=\nabla \log p^{\star}(t,x_t)$ denotes the score of the marginal density $p^{\star}(t,x_t)$ of the diffusion bridge process $X_t^{\star}$ at time $t\in(0,T)$. 
We refer readers to \citet{haussmann1986time}, \citet{millet1989integration}, and \citet{cattiaux2023time} for conditions under which the representation in \eqref{eqn:reversal_bridge} holds. By the Markov property, we have the relation
\begin{align}\label{eqn:marginal_bridge}
    p^{\star}(t,x_t) = p\{t,x_t\mid (0,x_0),(T,x_T)\} = \frac{p(t,x_t\mid 0,x_0)h(t,x_t)}{p(T,x_T\mid 0,x_0)},
\end{align}
as $h(t,x_t) = p(T,x_T\mid t,x_t)$. This implies that $s(t,x_t) = \nabla \log p(t,x_t\mid 0,x_0)$ is simply the score of the transition density of $X$.

Exploiting this backward time representation to derive diffusion bridge approximations is also highly non-trivial due to the intractability of the transition density $\eta(t,x_t)=p(t,x_t\mid 0,x_0)$, which is now characterized by the forward Kolmogorov equation $\partial_t\eta(t,x_t) + (\mathcal{F}\eta)(t,x_t)=0$, 
with initial condition at time $0$ given by the Dirac measure at $x_0$
and $\mathcal{F}$ denotes the Fokker--Planck operator of $X$ \citep{stroock1997multidimensional}. 
Numerical resolution of $\eta$ using partial differential equation solvers also suffers from the same difficulties as \eqref{eqn:backward_kolmogorov}. 
A key observation is that \eqref{eqn:reversal_bridge} can be understood as first setting $Z_0^{\star}=x_T$ to satisfy the terminal constraint, and then evolving $Z^{\star}$ using the time-reversal of \eqref{eqn:SDE}. Due to the influence of the score $s(t,x_t)$, the process will end at the initial constraint $Z_T^{\star}=x_0$ by construction. This connection between simulation of a diffusion bridge and time-reversal of its original diffusion process will form the basis of our score approximation. We refer readers to Section \ref{sec:supp_diffusion_bridge} of the Supplementary Material for an alternative and more elementary argument to establish this connection.

\subsection{Learning time-reversal with score matching}\label{sec:learn_reversal}
We introduce a variational formulation to learn the time-reversal of \eqref{eqn:SDE}, involving path measures on the space of continuous functions from $[0,T]$ to $\mathbb{R}^d$, equipped with the cylinder $\sigma$-algebra. 
Consider the time-reversed process $Z=(Z_t)_{t\in[0,T]}$ satisfying 
\begin{align}\label{eqn:reversed_process}
    \rmd Z_t = b_{\phi}(t,Z_t)\rmd t + \sigma(T-t,Z_t) \rmd B_t, 
    \quad Z_0\sim p(T,\rmd x_T\mid 0,x_0),
\end{align}
with drift function $b_{\phi}(t,z_t)=-f(T-t,z_t) + \Sigma(T-t,z_t)s_{\phi}(T-t,z_t) + \nabla \cdot\Sigma(T-t,z_t)$ that mimics the form of $b(t,z_t)$ in \eqref{eqn:reversal_bridge}. Here $s_{\phi}:[0,T]\times\mathbb{R}^d\rightarrow\mathbb{R}^d$ represents a function approximation of the score $s(t,x_t)$ that depends on parameters $\phi\in\Phi$ to be optimized. We shall measure the score approximation error as 
\begin{align}
\label{eq:e_phi_def}
    e(\phi) = {E}^{x_0}\left\lbrace\int_0^T\left\|s_{\phi}(t,X_t)-s(t,X_t)\right\|_{\Sigma(t,X_t)}^2\rmd t\right\rbrace,
\end{align}
where $E^{x_0}$ denotes expectation with respect to the path measure $\mathbb{P}^{x_0}$ induced by $(X_t)_{t\in[0,T]}$ in \eqref{eqn:SDE} with $X_0=x_0$, and $\|\cdot\|_{A}$ denotes the Euclidean norm weighted by a positive definite $A\in\mathbb{R}^{d\times d}$. 

Let $\mathbb{Q}_{\phi}^{x_0}$ be the path measure induced by the parameterised forward process $(Z_{T-t})_{t\in[0,T]}$. 
We consider minimizing the Kullback--Leibler divergence $\textsc{kl}(\mathbb{P}^{x_0}|\mathbb{Q}_{\phi}^{x_0}) = E^{x_0}\{\log \rmd \mathbb{P}^{x_0}/\rmd \mathbb{Q}_{\phi}^{x_0}(X)\}$, where $\rmd \mathbb{P}^{x_0}/\rmd \mathbb{Q}_{\phi}^{x_0}$ denotes the Radon--Nikodym derivative of $\mathbb{P}^{x_0}$ with respect $\mathbb{Q}_{\phi}^{x_0}$. 
The following result gives an upper bound of this objective and shows that the process $Z$ will end at the initial constraint $Z_T=x_0$ if the score approximation error is finite. 

\begin{proposition}
\label{prop:backward_KL}
Assuming $e(\phi)<\infty$, we have  $\textsc{kl}(\mathbb{P}^{x_0}|\mathbb{Q}_{\phi}^{x_0})\leq e(\phi) / 2$ and $Z_T=x_0$ holds $\mathbb{Q}_{\phi}^{x_0}$-almost surely. 
\end{proposition}
Under Novikov's condition, an exponential integrability assumption, Girsanov's theorem shows that $\textsc{kl}(\mathbb{P}^{x_0}|\mathbb{Q}_{\phi}^{x_0})= e(\phi) / 2$ \cite[Theorem 5.22, Theorem 4.13]{le2016brownian}. We do not provide here any assumption on $s_\phi(t,x_t)$ to guarantee that $e(\phi)<\infty$. 
Previous work by \citet{delyon2006simulation} and \citet{schauer2017guided} have studied necessary conditions for the law of a proposal bridge process and the law of the diffusion bridge to be equivalent. These findings suggest that the behaviour of $b_{\phi}(t,z_t)$ should be $(x_0-z_t)/t$ as $t\rightarrow 0$, which can be used to guide our parameterization of the score approximation. For example, we consider $s_{\phi}(t,x_t)=\Sigma^{-1}(t,x_t)(x_0-x_t)/t + v_{\phi}(t, x_t)$ where $v_{\phi}:[0,T]\times\mathbb{R}^d\rightarrow\mathbb{R}^d$ denotes a bounded function. 

Although Proposition \ref{prop:backward_KL} clearly relates the approximation of the time-reversal of \eqref{eqn:SDE} to the approximation of the score $s(t,x_t)$, its form is not amenable to optimization. The following result gives an alternative and practical expression by adapting the idea of denoising score matching \citep{vincent2011connection} to our setting. 

\begin{proposition}\label{prop:denoise}
For any partition $(t_m)_{m=0}^M$ of the interval $[0,T]$, we have $\textsc{kl}(\mathbb{P}^{x_0}|\mathbb{Q}_{\phi}^{x_0}) \leq L(\phi) + C$ if $e(\phi)<\infty$, where $C$ is a constant independent of $\phi\in\Phi$, the loss function is defined as 
\begin{align}\label{eqn:loss}
	L(\phi)=\frac{1}{2}\sum_{m=1}^M\int_{t_{m-1}}^{t_m}{E}^{x_0}\left\lbrace\left\|s_{\phi}(t,X_t) - g(t_{m-1},X_{t_{m-1}},t,X_t) \right\|_{\Sigma(t,X_t)}^2\right\rbrace\rmd t,
\end{align}
and $g(s,x_s,t,x_t)=\nabla\log p(t,x_t\mid s,x_s)$ for $0\leq s<t\leq T$ and $x_s,x_t\in\mathbb{R}^d$.
\end{proposition}

Therefore minimizing the Kullback--Leibler divergence $\textsc{kl}(\mathbb{P}^{x_0}|\mathbb{Q}_{\phi}^{x_0})$ is equivalent to minimizing the loss function $L(\phi)$. This allows us to circumvent the intractable score $s(t,x_t)$ by working with $g(t_{m-1},x_{t_{m-1}},t,x_t)$, the score of the transition density $p(t,x_t\mid t_{m-1},x_{t_{m-1}})$. Although the latter is also intractable, approximations can be made when the sub-interval $[t_{m-1},t_m]$ is sufficiently small. For example, under the Euler--Maruyama scheme \citep[p. 340]{kloeden1992} with stepsize $\delta t=T/M$, 
\begin{align*}
    g(t_{m-1},x_{t_{m-1}},t_m,x_{t_m})\approx-(\delta t)^{-1}\Sigma^{-1}(t_{m-1},x_{t_{m-1}})\{x_{t_m} - x_{t_{m-1}} - \delta t f(t_{m-1},x_{t_{m-1}})\}. 
\end{align*}
Hence the loss $L(\phi)$ can be approximated and minimized using stochastic gradient algorithms by simulating time-discretized paths under \eqref{eqn:SDE}. 
The minimal loss of $L(\phi)=-C$, achieved when $s_{\phi}(t,x_t)=s(t,x_t)$ $\mathbb{P}^{x_0}$-almost surely, is unknown as the constant $C$ is typically intractable. 
After obtaining the score approximation $s_{\phi}$, we can simulate a proposal bridge $Z$ from \eqref{eqn:reversed_process} with $Z_0=x_T$ and correct it using importance sampling or independent Metropolis--Hastings.
Time-discretization considerations and proposal correction procedures are detailed in Section \ref{sec:numerical_implement} of the Supplementary Material. 

In scenarios where one is interested in multiple pairs of conditioned states $(x_0,x_T)$, we can extend the above methodology to avoid having to learn multiple score approximations as follows. 
We let the score approximation $s_{\phi}(t,x_t)$ in \eqref{eqn:reversed_process} also depend on the initial condition $x_0$, and average the Kullback--Leibler objective $\textsc{kl}(\mathbb{P}^{x_0}|\mathbb{Q}_{\phi}^{x_0})$ with a distribution $p_0(\mathrm{d}x_0)$ on $\mathbb{R}^d$ that can be sampled from. By applying the arguments of Proposition~\ref{prop:denoise} conditionally on $X_0=x_0$, we obtain a loss function given by averaging \eqref{eqn:loss} over $p_0(\mathrm{d}x_0)$, which can be minimized using time-discretization and stochastic gradient algorithms. In this setting, we parameterize the score approximation as 
$s_{\phi}(t,x_t,x_0)=\Sigma^{-1}(t,x_t)(x_0-x_t)/t + v_{\phi}(t, x_t, x_0)$, where $v_{\phi}(t, x_t, x_0)$ is now a function that also depends on the initial condition $x_0$.

\subsection{Learning Doob's $h$-transform}\label{sec:doob_transform}
It is instructive to consider the time-reversal of \eqref{eqn:reversal_bridge}, which gives 
\begin{align}\label{eqn:double_reversal}
	\rmd X_t^{\star} = f^{\star}(t,X_t^{\star})\rmd t + \sigma(t,X_t^{\star})\rmd W_t, \quad X_0^{\star}=x_0,
\end{align}
with drift function $f^{\star}(t,x_t) = - b(T-t,x_t) + \Sigma(t,x_t)s^{\star}(t,x_t) + \nabla\cdot\Sigma(t,x_t)$. 
Using the form of $b(t,z_t)$ and the relation $s(t,x_t) = s^{\star}(t,x_t) - \nabla \log h(t,x_t)$, we can rewrite $f^{\star}$ as 
\begin{align}\label{eqn:recover_htransform}
	f^{\star}(t,x_t) = f(t,x_t) + \Sigma(t,x_t)\{s^{\star}(t,x_t)-s(t,x_t)\} = f(t,x_t) + \Sigma(t,x_t)\nabla\log h(t,x_t),
\end{align}
and recover the Doob's $h$-transform in \eqref{eqn:doob_SDE}. 
Although this is to be expected as the reversal of the time-reversed process should recover the original process, it forms the basis of our approximation of \eqref{eqn:doob_SDE}. 

After obtaining an approximation of $s(t,x_t)$, another iteration of our methodology can be used to obtain a function approximation of $s^{\star}(t,x_t)$. For brevity, this is detailed in Section \ref{sec:doob_transform} of the Supplementary Material. Plugging in both approximations in \eqref{eqn:recover_htransform} then gives an approximation of $X^{\star}$, which could be of interest in algorithms where one requires a forward time representation of the proposal bridge process \citep{lin2010generating,del2015sequential}. 
However, this comes at the cost of learning two approximations and some accumulation of errors. Recent work by \citet{baker2024score} have extended our approach to directly approximate $\nabla\log h(t,x_t)$, which is not necessarily the score of a probability density.

\section{Related work on generative modeling}
Denoising diffusion models are popular state-of-the-art generative models \citep{sohl2015deep,song2020score}. These models are based on a diffusion process that transforms data into random normal noise. Their time-reversal is then approximated to obtain a generative model that maps sampled noise to synthetic data. Recent extensions have generalized these models to allow the distribution at the terminal time $T$ to be arbitrary rather than a normal distribution \citep{de2021diffusion,vargas2021solving,chen2021likelihood}. While there are similarities between these methods and our proposed methodology, the challenges in simulating a diffusion bridge between two states $x_0$ and $x_T$ are distinctly different. Firstly, the diffusion process in \eqref{eqn:SDE} usually represents a probabilistic model for a problem of interest with an intractable transition density. In contrast, denoising diffusion models employ an Ornstein--Uhlenbeck process, leveraging its analytically tractable transition density to learn the score. 
Secondly, while the time dimension in \eqref{eqn:SDE} arises from modeling time series data, with the time interval $T$ representing the time between observations, the time variable $T$ in denoising diffusion models is artificially introduced and represents the time necessary for the diffusion to transform the data distribution into a normal distribution. Thirdly, the time-reversed process \eqref{eqn:reversal_bridge} is initialized from a conditioned state $x_T$ in our setting, whereas it is initialized from a normal distribution in generative models. 

\section{Numerical examples}\label{sec:examples}
\subsection{Preliminaries}\label{sec:preliminaries}
As our methodology allows one to employ any function approximator, we harness the flexibility of neural networks and the ease of implementation using modern software to approximate score functions. 
We adopt the parameterization of the score approximation in Section \ref{sec:learn_reversal} with $v_{\phi}$ defined by a neural network. The choice of neural network and stochastic gradient algorithm is detailed in Section \ref{sec:nn} of the Supplementary Material. Optimization times ranged between several seconds to a few minutes on a simple desktop machine and can be reduced with hardware accelerators. As such computational overheads are marginal when deploying proposal bridge processes within an independent Metropolis--Hastings algorithm with many iterations or an importance sampler with many samples, we focus on assessing the quality of our proposals in settings where existing proposal methods are unsatisfactory. The performance measures considered are the acceptance rate of independent Metropolis--Hastings, and in the case of an importance sampling estimator $\hat{p}(T,x_T\mid 0,x_0)$ of $p(T,x_T\mid 0,x_0)$, the effective sample size proportion, and the variance var$\{\log \hat{p}(T,x_T\mid 0,x_0)\}$ or the mean squared error ${E}[\{\log \hat{p}(T,x_T\mid 0,x_0)-\log p(T,x_T\mid 0,x_0)\}^2]$ when the true transition density is known. 
More implementation details are described in Section \ref{sec:implement_details} of the Supplementary Material. 
These measures were computed using $1024$ samples or iterations, and $100$ independent repetitions of each method. We benchmark our approximations of the backward diffusion bridge (BDB) and forward diffusion bridge (FDB) in \eqref{eqn:reversal_bridge} and \eqref{eqn:double_reversal} against 
the Clark--Delyon--Hu (CDH) proposal bridge studied by \citet{clark1990simulation} and \citet{delyon2006simulation}, the forward diffusion (FD) of \citet{pedersen1995consistency}, 
the guided proposal bridge (GPB) by \citet{schauer2017guided}, 
and the modified diffusion bridge (MDB) of \citet{durham2002numerical}. 
Given the difficulty of comparing the wide range of methods for diffusion bridges in a completely fair manner, as their strengths and weaknesses can depend on the specificities of the problem under consideration, we note that our objective is merely to illustrate a new avenue to improve the construction of proposal bridge processes. A Python package to reproduce all numerical results is available online\footnote{\url{https://github.com/jeremyhengjm/DiffusionBridge}}.

\subsection{Ornstein--Uhlenbeck process} 
Consider \eqref{eqn:SDE} with linear drift function $f(t,x_t)=-2 x_t$ and identity diffusion coefficient $\sigma(t,x_t)=I_d$. 
The transition density and score function are explicitly known and used as ground truth. 
The first and second rows of Fig. \ref{fig:ou} illustrates the impact of the time horizon $T$ and dimension $d$ on algorithmic performance with the constraints $x_0=x_T=(1,\ldots,1)$. 
To study how performance degrades when we condition on rare states under the diffusion process $\eqref{eqn:SDE}$, in the third row of Fig. \ref{fig:ou}, we set the initial state as $x_0=1$ and vary where the terminal state $x_T$ is in the tails of the transition density $p(T,x_T\mid 0,x_0)$ in multiples of its standard deviation.

We omit the guided proposal bridge in this example as it performs perfectly when its auxiliary process is the Ornstein--Uhlenbeck process. 
Our proposed methods offer substantial improvements over other methods without exploiting the correct parameterization. 
When comparing our forward and backward diffusion bridge approximations, we notice some accumulation of error, which is to be expected as the forward process is constructed using an additional score approximation. Given this observation, we will only consider the backward process in the following.

\begin{figure}
\centering
\includegraphics[width=1\linewidth]{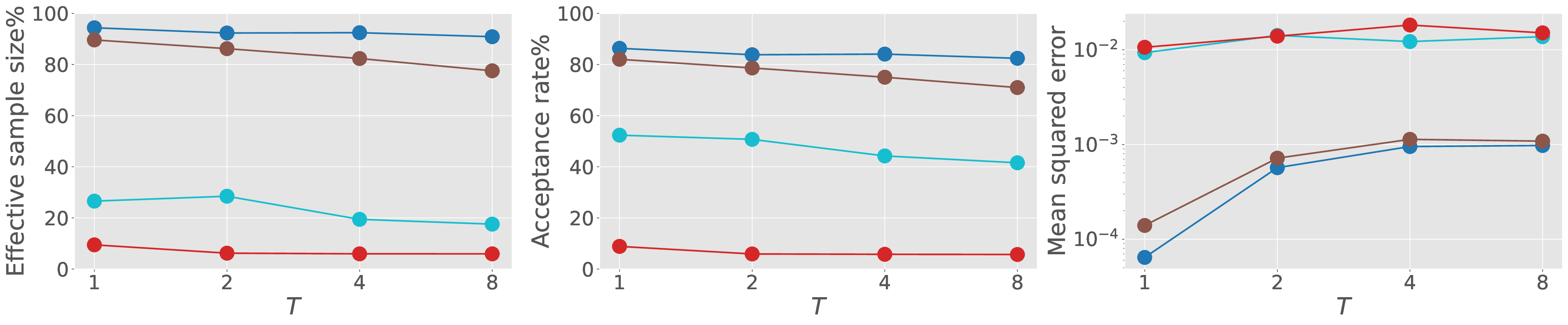}
\includegraphics[width=1\linewidth]{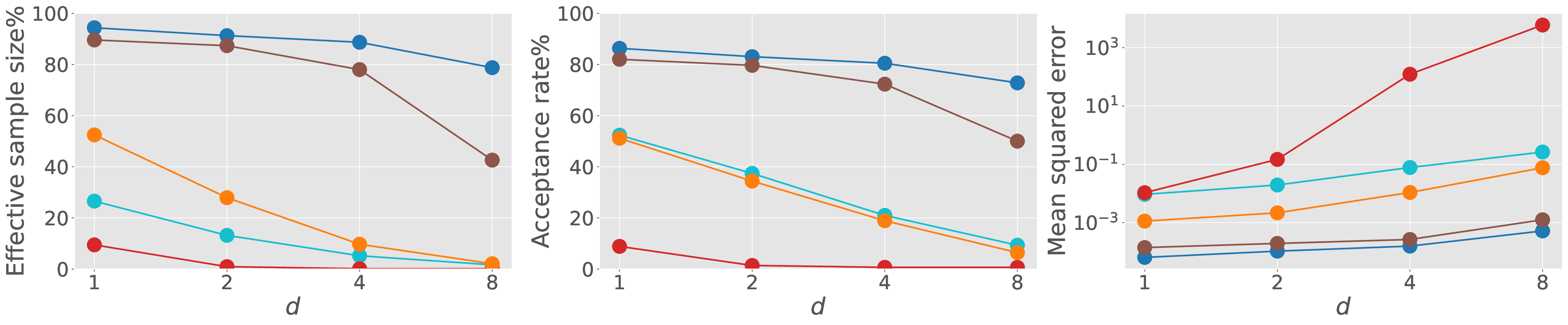}
\includegraphics[width=1\linewidth]{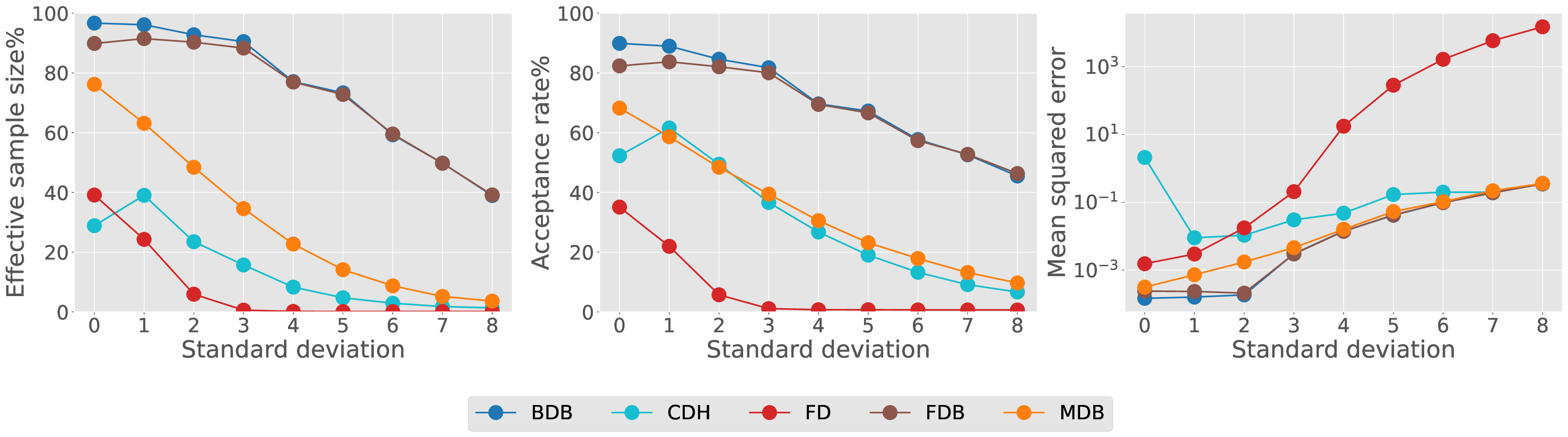}
\caption{Results for Ornstein--Uhlenbeck example with fixed $d=1$ (\emph{first row}), fixed $T=1$ (\emph{second row}), and fixed $d=1,T=1$ (\emph{third row}). 
}
\label{fig:ou}
\end{figure}

\subsection{Interest rates model} 
We consider a special case of an interest rates model in \citet{ait1998nonparametric}, defined by \eqref{eqn:SDE} with $f(t,x_t)=4/x_t-x_t$ and $\sigma(t,x_t)=1$. This specification admits a tractable transition density and score function as ground truth, given in terms of modified Bessel functions and its derivative. 
For each $T$, we first learn a single score approximation to handle multiple conditioned states $(x_0,x_T)$ by minimizing the loss in \eqref{eqn:loss} averaged over initial states $X_0=x_0$ from the gamma distribution with shape $5$ and rate $2$. 
We construct guided proposal bridges based on a first-order Taylor approximation of $f$ at the stable stationary point $x^{\star}=2$. 
We plot our score approximation error for $T=1$ in Fig. \ref{fig:radial_score}, and examine how algorithmic performance depends on $T$ and $(x_0,x_T)$ in Fig. \ref{fig:radial}. 
Our backward diffusion bridge approximation performs well for all $T$ considered and when conditioning requires the process to move away from the stationary point (first and third columns). Its performance is similar to guided proposal bridge and significantly better than other existing methods. 
While the performance of modified diffusion bridge typically degrades with $T$, the results of forward diffusion and Clark--Delyon--Hu bridge depend on the specific conditioned states. 

\begin{figure}
\centering
\includegraphics[width=1\linewidth]{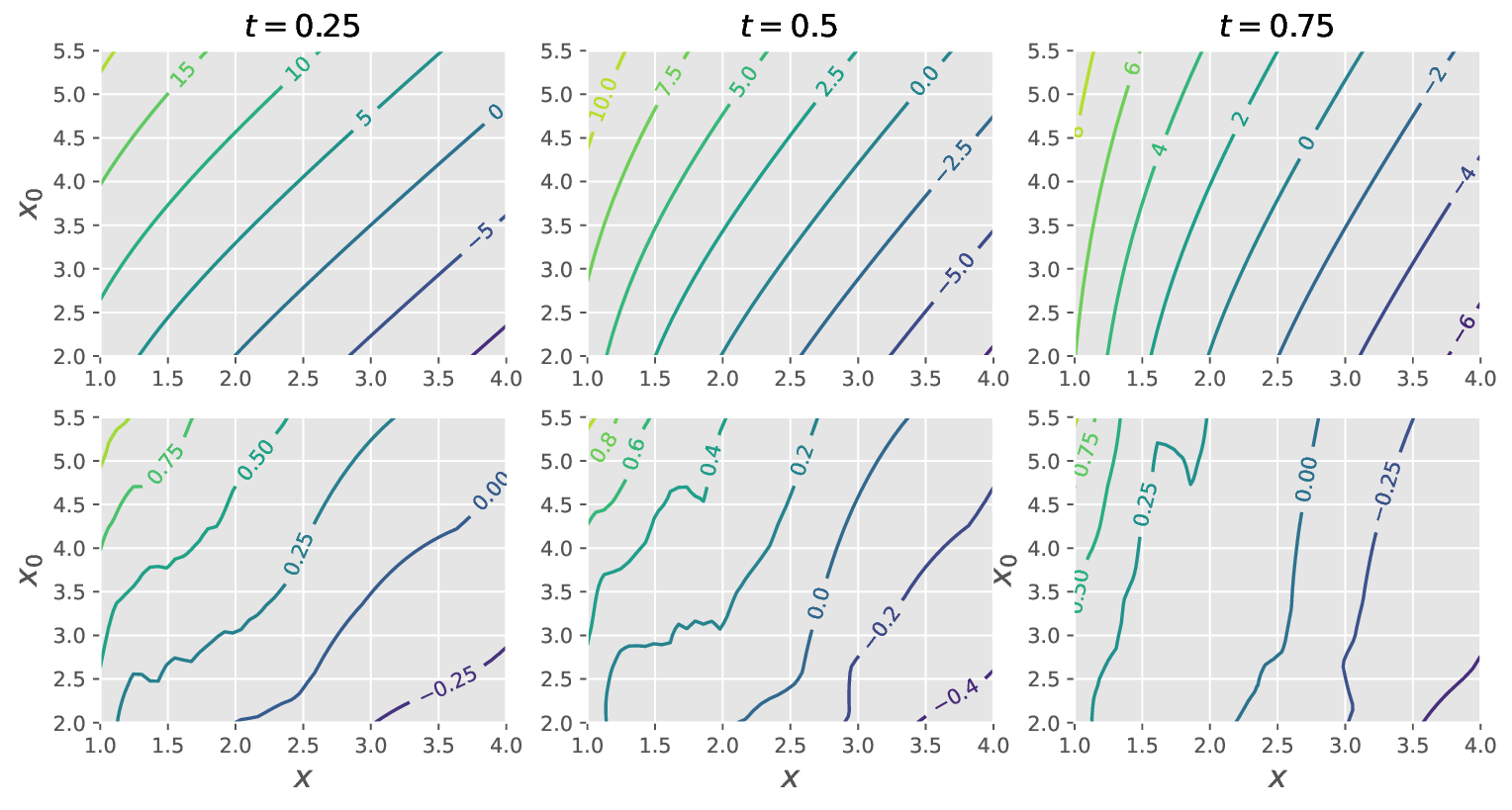}
\caption{True score function $s(t,x)$ (\emph{first row}) and the function approximation error $s(t,x)-s_{\phi}(t,x)$ (\emph{second row}) for various time $t$ and initial condition $x_0$.}
\label{fig:radial_score}
\end{figure}

\begin{figure}
\centering
\includegraphics[width=1\linewidth]{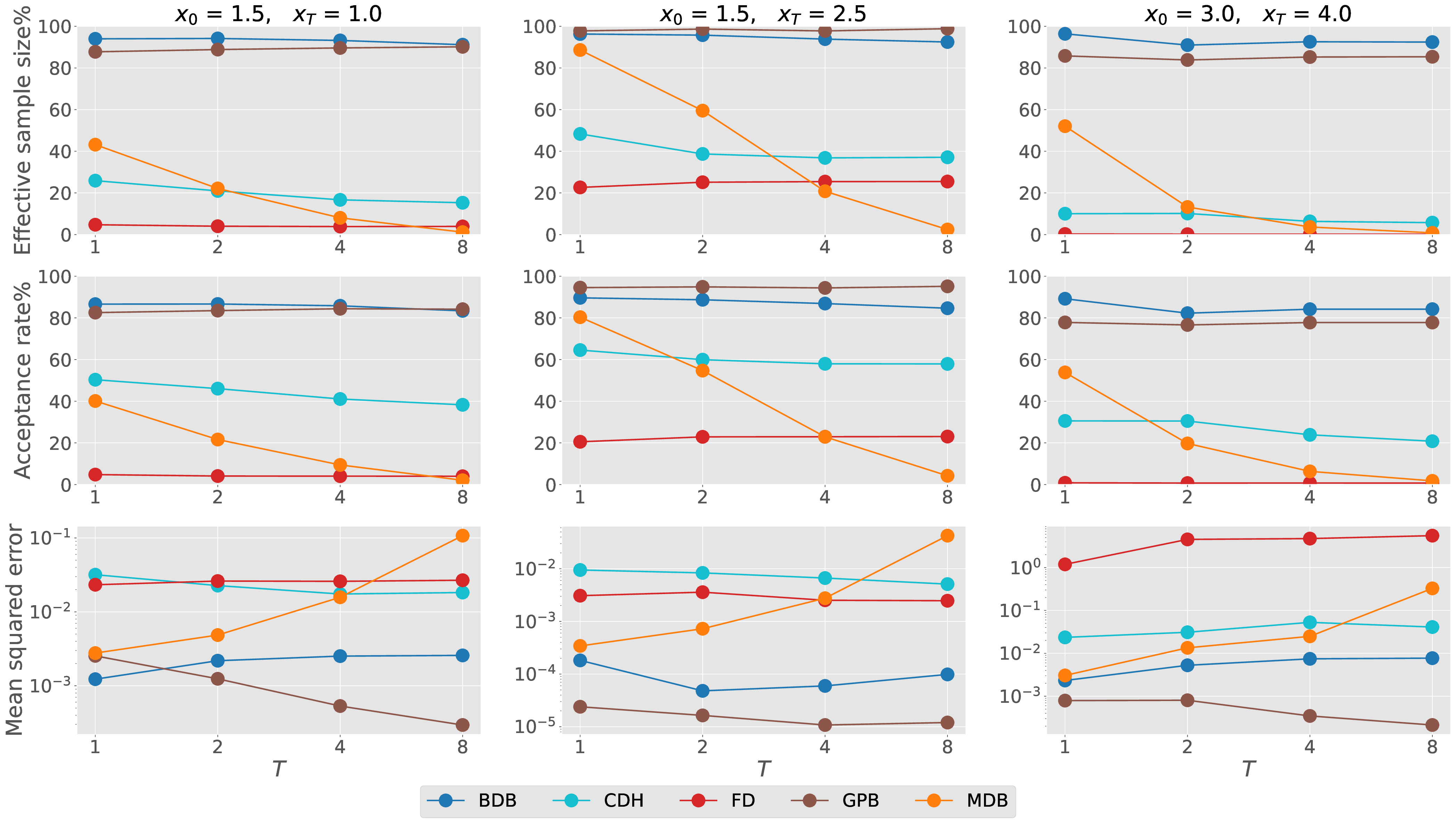}
\caption{Results for interest rates model for three pairs of conditioned states $(x_0,x_T)$.}
\label{fig:radial}
\end{figure}

\subsection{Cell model}
We end with a cell differentiation and development model from  \cite{wang2011quantifying}. Cellular expression levels $X_t=(X_{t,1},X_{t,2})$ of two genes are modelled by \eqref{eqn:SDE} with $f=(f_1,f_2)$, where 
$f_i(t,x_{t})=x_{t,i}^4/(2^{-4}+x_{t,i}^4) + 2^{-4}/(2^{-4}+x_{t,j}^4)- x_{t,i}$, 
for $(i,j)\in\{(1,2), (2,1)\}$, describe self-activation, mutual inhibition and inactivation respectively, and $\sigma(t,x_t)=\sigma_X I_2$ captures intrinsic or external fluctuations. We consider the cell development from the undifferentiated state of $x_0=(1,1)$ to a differentiated state $x_T=(x_{T,1},x_{T,2})$ defined by another stable fixed point of $f$ satisfying $x_{T,1}<x_{T,2}$. 
We employ an auxiliary Ornstein--Uhlenbeck process in the guided proposal bridge and optimize its parameters using the Kullback--Leibler divergence considered in \citet{schauer2017guided}. 
Fig. \ref{fig:cell} displays the performance of all proposal methods to approximate this diffusion bridge. For this application, we observe improvement over guided proposal bridge when the diffusion coefficient $\sigma_X$ is smaller, and significant improvement over other methods when the time horizon is long. These findings are consistent with numerical results obtained in \citet{baker2024score} for the same model.

\begin{figure}
\centering
\includegraphics[width=1\linewidth]{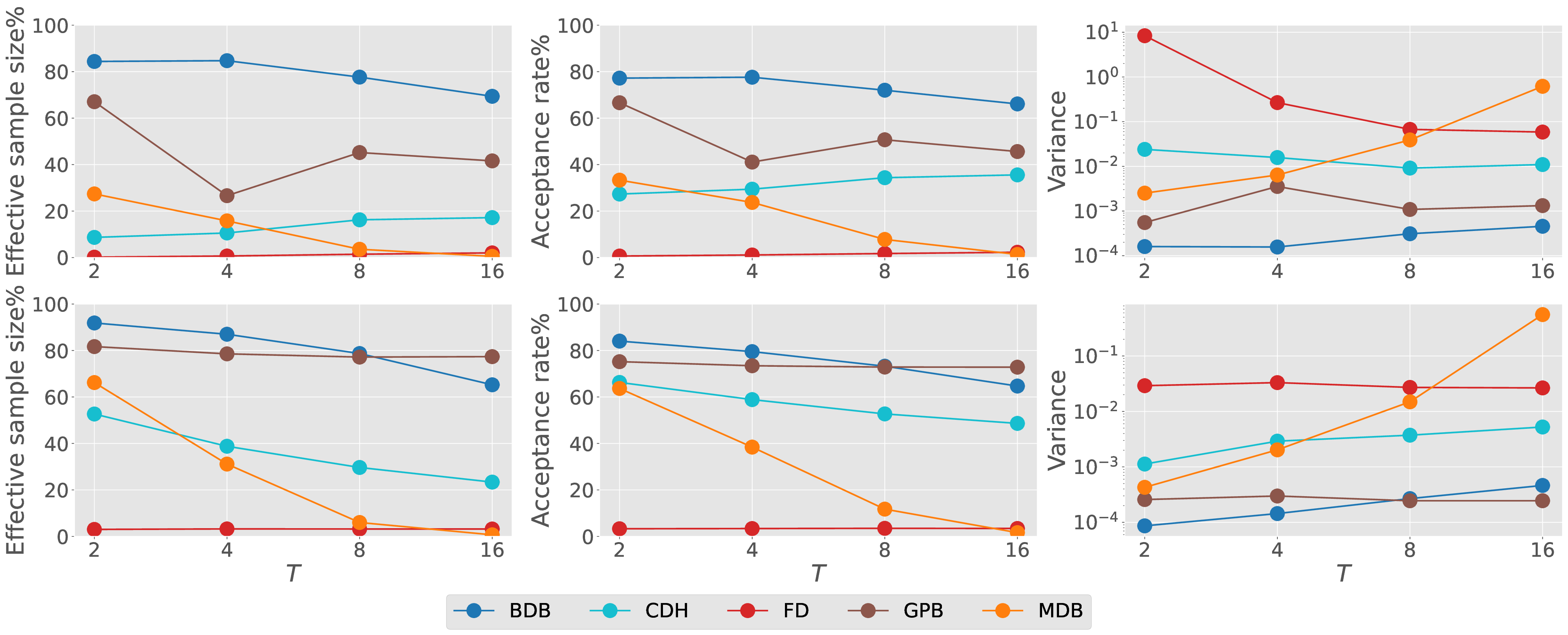}
\caption{Results for cell model with $\sigma_X^2=0.1$ (\emph{first row}) and $\sigma_X^2=1$ (\emph{second row}).}
\label{fig:cell}
\end{figure}
    
\section{Discussion}
While implicit score matching \citep{hyvarinen2005estimation} could have been an alternative to denoising score matching \citep{vincent2011connection}, we found it to be less effective and computationally more expensive. 
If the terminal state $x_T$ is highly unlikely under the law of the unconditional diffusion, the score approximation learned using unconditional sample paths could be poor near $x_T$. Nevertheless, we observed satisfactory experimental results for an Ornstein--Ulhenbeck process up to 8 standard deviations. 
In statistical problems where some hyperparameters of the diffusion are estimated using maximum likelihood or Bayesian inference, a direct application of our approach would require a new diffusion bridge approximation for each set of parameters considered. To avoid incurring significant computational cost, one could also incorporate parameter dependence in the score approximation procedure \citep{boserup2024parameter}. 
Recent work has also extended our approach to approximate directly $\nabla \log h(t,x_t)$ without having to approximate a time-reversal \citep{baker2024score} and to infinite-dimensional diffusion processes \citep{baker2024conditioning}.

\section*{Acknowledgement}
The authors thank the editor, associate editor, and reviewers for their comments, which have helped to improve the paper significantly. 
Arnaud Doucet and James Thornton were partially supported by the U.K. Engineering and Physical Sciences Research Council.

\bibliography{ref}

\appendix

\section{Proofs of Propositions~\ref{prop:backward_KL} and \ref{prop:denoise}}\label{sec:proofs_backward}
In the following, we write $\langle x, y\rangle_{A}=x^\T A y$ for the inner product of $x,y\in\mathbb{R}^d$ weighted by a positive definite matrix $A\in\mathbb{R}^{d\times d}$ and $\|x\|_{A}=\sqrt{\langle x, x\rangle_{A}}$ for the induced weighted Euclidean norm. 
Before diving into the proof of Proposition~\ref{prop:backward_KL}, we recall the data processing inequality. A proof of such result can be found in \citet[Lemma 9.4.5]{ambrosio2005gradient}.

\begin{lemma}
\label{prop:data_processing}
Assume that $\pi_0, \pi_1$ are two probability distributions over a metric space $\mathcal{X}$. Let $\varphi: \mathcal{X} \to \mathcal{X}$ be a Borel map, i.e. a measurable map for the $\sigma$-algebra generated by the open sets of $\mathcal{X}$. Then, we have that 
\begin{equation*}
    \mathrm{KL}(\varphi_{\#} \pi_0 | \varphi_{\#} \pi_1 ) \leq \mathrm{KL}(\pi_0 | \pi_1) . 
\end{equation*}
\end{lemma}
We are now ready to prove the main result of this section. 

\begin{proof}[Proof of Proposition~\ref{prop:backward_KL}]
The proof is similar to \citet[Section 5.2]{chen2022sampling}, see also \cite{benton2023linear}.  We recall that a process $(B_t)_{t \in [0,T]}$ defined in $(\Omega, \mathcal{F}, (\mathcal{F}_t)_{t \geq 0}, \mathbb{P})$ is a $\mathbb{P}$-Brownian motion if $(B_t)_{t \in [0,T]}$ is almost surely continuous and for any $s, t \in [0,T]$, $B_t - B_s$ is a normal random variable with zero mean and covariance matrix $|t-s| I_d$. If there is no possible ambiguity, we simply say that $(B_t)_{t \in [0,T]}$ is a Brownian motion. First, we give a version of  Girsanov's theorem which is a consequence of \citet[Theorem 4.13, Theorem 5.22]{legall2016brownian}.

    Let us give a few more details on the applications of \citet[Theorem 4.13]{legall2016brownian} and \citet[Theorem 5.22]{legall2016brownian} for the derivation of Theorem~\ref{prop:standardGirsanov}. We follow the approach described in \citet[p.136]{legall2016brownian}, where a process $\mathcal{L}_t = \int_0^t \Delta_s \mathrm{d} B_s$ is first defined. Using \citet[Theorem 4.13]{legall2016brownian}, we immediately get that the process $(\mathcal{L}_t)_{t \in [0,T}$ is a martingale. We follow the rest of the paragraph of \citet[p.136]{legall2016brownian}, notably leveraging that ${E}_{\mathbb{P}}[\mathcal{E}(\mathcal{L})_T] = 1$, to obtain that $(\mathcal{E}(\mathcal{L})_{t \in [0,T]}$ is a $\mathbb{P}$-martingale. The rest of the theorem is a direct application of \citet[Theorem 5.22]{legall2016brownian}.

\begin{theorem}
\label{prop:standardGirsanov}
Let $(\Omega, \mathcal{F}, (\mathcal{F}_t)_{t \geq 0}, \mathbb{P})$ be a filtered probability space and $(\Delta_t)_{t \in [0,T]}$ be a predictable process with ${E}_{\mathbb{P}}\big[\int_0^T \| \Delta_s \|^2 \rmd s\big] < \infty$. Let $(B_t)_{t \geq 0}$ be a $\mathbb{P}$-Brownian motion and define $\mathcal{L}_t = \int_0^t \Delta_s \rmd B_s$. Then, $(\mathcal{L}_t)_{t \in [0,T]}$ is a $\mathbb{P}$-martingale. For any $t \in [0,T]$, we define 
\begin{equation*}
    \mathcal{E}(\mathcal{L})_t = \exp\Big(\int_0^t \Delta_s \rmd B_s - \frac{1}{2} \int_0^t \|\Delta_s\|^2 \rmd s\Big) .
\end{equation*}
Assume that for any $t \in [0,T]$, ${E}_{\mathbb{P}}[\mathcal{E}(\mathcal{L})_T] = 1$ then $(\mathcal{E}(\mathcal{L})_t)_{t \in [0,T]}$ is a $\mathbb{P}$-martingale. Let $\mathbb{Q}$ be a probability measure such that $\rmd \mathbb{Q}/\rmd \mathbb{P} = \mathcal{E}(\mathcal{L})$ and let $(\beta_t)_{t \in [0,T]}$ such that for any $t \in [0,T]$
\begin{equation*}
    \beta_t = B_t - \int_0^t \Delta_s \rmd s .
\end{equation*}
Then $(\beta_t)_{t \in [0,T]}$ is a $\mathbb{Q}$-Brownian motion.
\end{theorem}

Let $x_0 \in \mathbb{R}^d$ and we recall that  $\mathbb{Q}_{\phi}^{x_0}$ is the path measure induced by a time-reversed process $Z=(Z_t)_{t\in[0,T]}$. More precisely we have that $\mathbb{Q}_{\phi}^{x_0}$ is the distribution of $(Z_{T-t})_{t \in [0,T]}$, where $Z=(Z_t)_{t\in[0,T]}$ satisfies 
\begin{align}
    \rmd Z_t = b_{\phi}(t,Z_t)\rmd t + \sigma(T-t,Z_t) \rmd B_t, 
    \quad Z_0\sim p(T, \rmd x_T\mid 0,x_0),
\end{align}
with drift function $b_{\phi}(t,z_t)=-f(T-t,z_t) + \Sigma(T-t,z_t)s_{\phi}(T-t,z_t) + \nabla \cdot\Sigma(T-t,z_t)$. In addition, we have that $Z^{\star}=(Z_t^{\star})_{t\in[0,T]}=(X_{T-t})_{t\in[0,T]}$ is associated with $\mathbb{P}^{x_0}$. More precisely $\mathbb{P}^{x_0}$ is the distribution of $(X_t)_{t \in [0,T]} = (Z_{T-t}^\star)_{t \in [0,T]}$ which satisfies
\begin{align}
    \rmd Z_t^{\star} = b(t,Z_t^{\star})\rmd t + \sigma(T-t,Z_t^{\star}) \rmd B_t', \quad Z_0^\star\sim p(T,\rmd x_T\mid 0,x_0),
\end{align}
with drift function $b(t,z_t) = -f(T-t,z_t) + \Sigma(T-t,z_t)s(T-t,z_t) + \nabla \cdot\Sigma(T-t,z_t)$. Here $(B'_t)_{t\in[0,T]}$ defines a $d$-dimensional Brownian motion. In what follows, we define $(\Delta_t)_{t \in [0,T]}$ such that for any $t \in [0,T]$ 
\begin{equation*}
 \Delta_t = \sigma(T-t, Z_t^\star)^\T \{s(T-t, Z_t^\star) - s_\phi(T-t, Z_t^\star)\} .
\end{equation*}
We recall that 
\begin{align}
\label{eq:e_phi_def_supp}
    e(\phi) &= {E}^{x_0}\left\lbrace\int_0^T\left\|s_{\phi}(t,X_t)-s(t,X_t)\right\|_{\Sigma(t,X_t)}^2\rmd t\right\rbrace \\
    &= {E}^{x_0}\left\lbrace\int_0^T\left\|s_{\phi}(T-t,Z_t^\star)-s(T-t,Z_t^\star)\right\|_{\Sigma(T-t,Z_t^\star)}^2\rmd t\right\rbrace,
\end{align}
where $E^{x_0}$ denotes expectation with respect to $\mathbb{P}^{x_0}$. 

We now provide an outline for the rest of the proof. First, we show that we can define a sequence of measures which approximate $\mathbb{Q}_\phi^{x_0}$. This sequence, denoted by $(\mathbb{Q}_\phi^{n})_{n \in \mathbb{N}}$, corresponds to changing the behaviour of the approximate process, see \eqref{eq:process_approx_n}. In particular, near time $T$, we use the true score $s$ instead of the approximate score $s_\phi$. The important property of this sequence of measures is that we can derive a uniform bound on the Kullback--Leibler divergence $\textsc{KL}(\mathbb{P}^{x_0} |\mathbb{Q}^n_\phi)$, see \eqref{eq:KLQPnbound}. This is done by leveraging results from Girsanov theory, see Theorem~\ref{prop:standardGirsanov}. Finally, we need to ensure that this result implies that $\textsc{KL}(\mathbb{P}^{x_0} |\mathbb{Q}^{x_0}_\phi)$ is bounded. It can easily be seen that when truncating the time, i.e. considering $\mathrm{proj}_\varepsilon(\omega)(t) = \omega(t \wedge (T-\varepsilon))$, we have that for any $\varepsilon > 0$, there exists $n_0 \in \mathbb{N}$ such that for any $n \in \mathbb{N}$ with $n \geq n_0$,  $(\mathrm{proj}_\varepsilon)_{\#} \mathbb{Q}_\phi^n = (\mathrm{proj}_\varepsilon)_{\#} \mathbb{Q}_\phi^{x_0}$. Therefore, using lower semi-continuity of the Kullback--Leibler divergence, see \citet[Lemma 9.4.3]{ambrosio2005gradient}, we are able to provide a finite bound for $\textsc{KL}((\mathrm{proj}_\varepsilon)_{\#}\mathbb{P}^{x_0} |(\mathrm{proj}_\varepsilon)_{\#}\mathbb{Q}^{x_0}_\phi)$ which is uniform in $\varepsilon$. Finally, we can use the data processing inequality in Lemma~\ref{prop:data_processing} to conclude the proof.

Note that $(\Delta_t)_{t \in [0,T]}$ is a predictable process and using the assumption that $ e(\phi) < +\infty$, we have
\begin{equation*}
    {E}^{x_0}\Big[\int_0^{T} \| \Delta_t \|^2 \rmd t\Big]  < + \infty . 
\end{equation*}
For any $t \in [0,T]$, let $\mathcal{L}_t = \int_0^t \Delta_s \rmd B'_s$. Then, using \citet[Proposition 5.11]{legall2016brownian} we have that $(\mathcal{E}(\mathcal{L})_t)_{t \in [0,T]}$ is a continuous local martingale. As a result there exists a  sequence of stopping times $(T_n)_{n \in \mathbb{N}}$ such that $T_n \rightarrow T$ almost surely and $(\mathcal{E}(\mathcal{L})_{t \wedge T_n})_{t \in [0,T]}$ is a continuous martingale for each $n$.

For any $n \in \mathbb{N}$, let  $\mathcal{L}_t^n = \int_0^t \Delta_s 1_{[0,T_n]}(s) \rmd B_s'$ for all $t \in [0,T]$ and $n \in \mathbb{N}$, where $1_A(s)=0$ if $s\in A$ and $0$ if $s\notin A$. Then for any $t \in [0,T]$, we have $\mathcal{E}(\mathcal{L})_{t \wedge T_n} = \mathcal{E}(\mathcal{L}^n)_t$, so $\mathcal{E}(\mathcal{L}^n)$ is a continuous martingale, and it follows that ${E}^{x_0}[\mathcal{E}(\mathcal{L}^n)_{T}] = 1$. Hence, using Theorem \ref{prop:standardGirsanov} we have that, for any $n \in \mathbb{N}$, $\mathbb{Q}^n_\phi$ is such that $\rmd \mathbb{Q}^n_\phi/\rmd \mathbb{P}^{x_0} = \mathcal{E}(\mathcal{L}^n)_{T}$ is a probability measure. Moreover, for any $n \in \mathbb{N}$, $(\beta^n_t)_{t \in [0,T]}$ given by  
\begin{equation*}
    \beta^n_t = B_t' - \int_0^t \Delta_s 1_{[0,T_n]}(s) \rmd s 
\end{equation*}
is a $\mathbb{Q}^n_\phi$-Brownian motion. We also have
\begin{align*}
    \rmd Z_t^\star &= \{ -f(T-t, Z_t^\star) + \Sigma(T-t, Z_t^\star)s_\phi(T-t, Z_t^\star) + \nabla \cdot\Sigma(T-t,Z_t^\star) 1_{[0,T_n]}(t) \rmd t \\
    & \qquad + \{ -f(T-t, Z_t^\star) + \Sigma(T-t, Z_t^\star)s(T-t, Z_t^\star)\} 1_{(T_n, T]}(t) \rmd t + \sigma(T-t, Z_t^\star) \rmd \beta^n_t.
\end{align*}

To clarify, this can also be rewritten as \begin{align*}
    \rmd Z_t^\star &= \{ b_\phi(t, Z_t^\star) 1_{[0,T_n]}(t) + b(t, Z_t^\star) 1_{(T_n,T]}(t) \} \rmd t + \sigma(T-t, Z_t^\star) \rmd \beta^n_t.
\end{align*}
In addition, we have that 
\begin{align}
    \textsc{KL}(\mathbb{P}^{x_0} |\mathbb{Q}^n_\phi) & = {E}^{x_0}[\log \rmd \mathbb{P}^{x_0}/\rmd \mathbb{Q}^n_\phi] = - {E}^{x_0}[\log \mathcal{E}(\mathcal{L}^n)_{T}] \nonumber \\
    & = {E}^{x_0}\left[ - \mathcal{L}_{T_n} + \frac{1}{2} \int_0^{T_n} \|\Delta_s\|^2 \rmd s \right] = \frac{1}{2}{E}^{x_0}\left[\int_0^{T_n} \|\Delta_s\|^2 \rmd s \right] \label{eq:KLQPnbound}
\end{align}
as $\mathcal{L}$ is a $\mathbb{P}^{x_0}$-martingale. Finally, we define for any $n \in \mathbb{N}$ 
\begin{align*}
    \rmd Z_t^n &= \{ -f(T-t, Z_t^n) + \Sigma(T-t, Z_t^n)s_\phi(T-t, Z_t^n) + \nabla \cdot\Sigma(T-t,Z_t^\star) \} 1_{[0,T_n]}(t) \rmd t \\
    & \qquad + \{ -f(T-t, Z_t^n) + \Sigma(T-t, Z_t^n)s(T-t, Z_t^n)\} 1_{(T_n, T]}(t) \rmd t + \sigma(T-t, Z_t^\star)\rmd W_t,
\end{align*}
with $(W_t)_{t \in [0,T]}$ a $d$-dimensional Brownian motion defined over a filtered probability space $(\Omega, \mathcal{F}, (\mathcal{F}_t)_{t \geq 0}, \mathbb{M})$ and $Z_0^n \sim p(T, \rmd x_T | 0, x_0) $. For any $n \in \mathbb{N}$, we have that $\mathbb{Q}^n_\phi$ is the probability measure associated with $(Z_t^n)_{t \in [0,T]}$.
We also define 
\begin{align*}
    \rmd Z_t^{\phi} &= \{ -f(T-t, Z_t^{\phi}) + \Sigma(T-t, Z_t^{\phi})s_\phi(T-t, Z_t) + \nabla \cdot\Sigma(T-t,Z_t^{\phi}) \}  \rmd t  + \sigma(T-t, Z_t^{\phi}) \rmd W_t,
\end{align*}
and $Z_0 \sim p(T, \rmd x_T | 0, x_0) $. Note that $(Z_t^{\phi})_{t \in [0,T]}$ is associated with $\mathbb{Q}_\phi^{x_0}$.
Let $\varepsilon > 0$ and define $\mathrm{proj}_\varepsilon : \mathcal{C}([0,T], \mathbb{R}^d) \rightarrow \mathcal{C}([0,T], \mathbb{R}^d)$ by $\mathrm{proj}_\varepsilon(\omega)(t) = \omega(t \wedge (T - \varepsilon))$ for $t \in [0,T]$. Then, for any $\varepsilon > 0$, $\mathrm{proj}_\varepsilon(Z^n) \rightarrow \mathrm{proj}_\varepsilon(Z^\phi)$ uniformly over $[0,T]$ almost surely. Therefore, using  \citet[Lemma 12]{chen2022sampling}, we have  $(\mathrm{proj}_\varepsilon)_\# \mathbb{Q}^n_\phi \rightarrow (\mathrm{proj}_\varepsilon)_\# \mathbb{Q}^{x_0}_\phi$. Using \citet[Lemma 9.4.3]{ambrosio2005gradient}, the data processing inequality, see Lemma~\ref{prop:data_processing}, and \eqref{eq:KLQPnbound}, we get 
\begin{align}
    \textsc{KL}((\mathrm{proj}_\varepsilon)_\# \mathbb{P}^{x_0}|(\mathrm{proj}_\varepsilon)_\# \mathbb{Q}^{x_0}_\phi) & \leq \liminf_{n \rightarrow \infty} \textsc{KL}((\mathrm{proj}_\varepsilon)_\# \mathbb{P}^{x_0} | (\mathrm{proj}_\varepsilon)_\# \mathbb{Q}^n_\phi) \nonumber \\
    & \leq \liminf_{n \rightarrow \infty} \textsc{KL}(\mathbb{P}^{x_0} | \mathbb{Q}^n_\phi)  \leq e(\phi) < +\infty . \label{eq:process_approx_n}
\end{align}
Finally, letting $\varepsilon \rightarrow 0$ we have that $\mathrm{proj}_\varepsilon(\omega) \rightarrow \omega$ uniformly on $[0,T]$ \citep[Lemma 13]{chen2022sampling}, and hence using \citet[Corollary 9.4.6]{ambrosio2005gradient}, we get that $\textsc{KL}((\mathrm{proj}_\varepsilon)_\# \mathbb{P}^{x_0} | (\mathrm{proj}_\varepsilon)_\# \mathbb{Q}^{x_0}_\phi) \rightarrow \textsc{KL}(\mathbb{P}^{x_0} | \mathbb{Q}^{x_0}_\phi)$. Hence $\textsc{KL}(\mathbb{P}^{x_0}|\mathbb{Q}_{\phi}^{x_0})<\infty$ under the assumption that $e(\phi)<\infty$. 
\end{proof}

\begin{proof}[Proof of Proposition~\ref{prop:denoise}]
By expanding the square in \eqref{eq:e_phi_def_supp}, we can decompose the upper bound on the Kullback--Leibler divergence as 
\begin{align}\label{eqn:KL_decomp}
	\textsc{kl}(\mathbb{P}^{x_0}|\mathbb{Q}_{\phi}^{x_0}) \leq C_1 + C_2 -  C_3,
\end{align}
where 
\begin{align}
    C_1 &= \frac{1}{2}\int_0^T\left\|s(t,x_t)\right\|_{\Sigma(t,x_t)}^2p(t,x_t\mid 0,x_0)\rmd x_t\rmd t,\notag\\
    C_2 &= \frac{1}{2}\sum_{m=1}^M\int_{t_{m-1}}^{t_m}\int_{\mathbb{R}^d} \| s_{\phi}(t,x_t)\|_{\Sigma(t,x_t)}^2 p(t,x_t\mid 0,x_0)\rmd x_t\rmd t,\label{eqn:D2}\\
    C_3 &= \sum_{m=1}^M\int_{t_{m-1}}^{t_m}\int_{\mathbb{R}^d} \left\langle s_{\phi}(t,x_t), s(t,x_t)\right\rangle_{\Sigma(t,x_t)} p(t,x_t\mid 0,x_0)\rmd x_t \rmd t.\notag
\end{align}
We examine the term $C_3$ that depends on the unknown score function $s(t,x_t)$. Firstly, we can write 
\begin{align*}
	C_3 = \sum_{m=1}^M\int_{t_{m-1}}^{t_m}\int_{\mathbb{R}^d} \left\langle s_{\phi}(t,x_t), \nabla p(t,x_t\mid 0,x_0) \right\rangle_{\Sigma(t,x_t)} \rmd x_t\rmd t.
\end{align*}
By differentiating the Chapman--Kolmogorov equation with respect to the variable $x_t$
\begin{align*}
	\nabla p(t,x_t\mid 0,x_0) &= \nabla \int_{\mathbb{R}^d} p(t,x_t\mid t_{m-1},x_{t_{m-1}}) p(t_{m-1},x_{t_{m-1}}\mid 0,x_0)\rmd x_{t_{m-1}}\\
	&= \int_{\mathbb{R}^d} \nabla\log p(t,x_t\mid t_{m-1},x_{t_{m-1}}) p(t,x_t\mid t_{m-1},x_{t_{m-1}}) p(t_{m-1},x_{t_{m-1}}\mid 0,x_0) \rmd x_{t_{m-1}},\notag	
\end{align*}
we obtain
\begin{align}\label{eqn:D3}
	C_3 = \sum_{m=1}^M\int_{t_{m-1}}^{t_m}{E}^{x_0}\left\lbrace \left\langle s_{\phi}(t,X_t), \nabla\log p(t,X_t\mid t_{m-1},X_{t_{m-1}}) \right\rangle_{\Sigma(t,X_t)} \right\rbrace  \rmd t.
\end{align}
By expanding the square in \eqref{eqn:loss} and using \eqref{eqn:D2}, \eqref{eqn:D3}, and 
\begin{align*}
    C_4 &= \frac{1}{2}\sum_{m=1}^M\int_{t_{m-1}}^{t_m} {E}^{x_0}\left\lbrace\|g(t_{m-1},X_{t_{m-1}},t,X_t)\|_{\Sigma(t,X_t)}^2\right\rbrace\rmd t,
\end{align*}
we have 
\begin{align*}
	L(\phi) = C_2 - C_3 + C_4.
\end{align*}
The claim follows by noting the decomposition in \eqref{eqn:KL_decomp} and taking $C=C_1-C_4$.
\end{proof}

\section{Diffusion bridges}\label{sec:supp_diffusion_bridge}
\subsection{Time-reversed bridge process}
Here we provide an alternative way to establish that the time-reversed bridge process $(Z_t^{\star})_{t\in[0,T]}=(X_{T-t}^{\star})_{t\in[0,T]}$ evolves according to the time-reversal of the original diffusion process $(X_t)_{t\in[0,T]}$ in  \eqref{eqn:SDE} with initialization at $Z_0^{\star}=x_T$. For any $M\geq 3$, let $0=t_0<t_1<\cdots<t_M=T$ denote a partition of the interval $[0,T]$. The finite-dimensional distribution of $(X_{t_m}^{\star})_{m=1}^{M-1}$ is equal to that of $(X_{t_m})_{m=1}^{M-1}$ conditioned on $X_0=x_0$ and $X_T=x_T$, which can be written as
\begin{align}
    P^{x_0,x_T}(\rmd x_{t_1},\ldots,\rmd x_{t_{M-1}}) 
    = P^{x_0,x_T}(\rmd x_{t_{M-1}}) \prod_{m=1}^{M-2} P^{x_0,x_T}(\rmd x_{t_m} \mid x_{t_{m+1}},\ldots, x_{t_{M-1}}).
\end{align}
We have 
\begin{align}
    P^{x_0,x_T}(\rmd x_{t_{M-1}}) = 
    \frac{p(T, x_T \mid t_{M-1}, x_{t_{M-1}}) p(t_{M-1}, x_{t_{M-1}} \mid 0, x_0) \rmd x_{t_{M-1}} }
    {p(T,x_T\mid 0, x_0)},
\end{align}
and for each $m\in\{1,\ldots,M-2\}$
\begin{align}
    P^{x_0,x_T}(\rmd x_{t_m} \mid x_{t_{m+1}},\ldots, x_{t_{M-1}}) 
    = \frac{p(t_{m+1}, x_{t_{m+1}} \mid t_m, x_{t_m}) p(t_m, x_{t_m}\mid 0, x_0) \rmd x_{t_m}}{p(t_{m+1}, x_{t_{m+1}}\mid 0, x_0)},
\end{align}
which are indeed the transition kernels of the time-reversed process $(X_{t_{M-m}})_{m=1}^{M-1}$.

\subsection{Learning Doob's $h$-transform}
Suppose we have found a minimizer $\hat{\phi}\in\arg\min_{\phi\in\Phi}L(\phi)$ satisfying $e(\hat{\phi})<\infty$ and denote the corresponding score approximation as 
$\hat{s}(t,x_t)=s_{\hat{\phi}}(t,x_t)$ and drift function as $\hat{b}(t,z_t)=b_{\hat{\phi}}(t,z_t)$. 
Consider a time-reversed bridge process $\hat{Z}=(\hat{Z}_t)_{t\in[0,T]}$ satisfying
\begin{align}\label{eqn:approx_reverse}
	\rmd \hat{Z}_t = \hat{b}(t,\hat{Z}_t)\rmd t + \sigma(T-t,\hat{Z}_t) \rmd B_t, \quad \hat{Z}_0 = x_T,
\end{align}
which should be seen as an approximation of \eqref{eqn:reversal_bridge}. 
Let $\hat{q}(t,z_t\mid s,z_s)$ denote its transition density for any $0\leq s<t\leq T$, $\hat{\mathbb{Q}}^{x_0,x_T}$ be the induced path measure, and $\hat{E}^{x_0,x_T}$ to denote expectation with respect to $\hat{\mathbb{Q}}^{x_0,x_T}$. Note that $\hat{p}^{\star}(t,x_t)=\hat{q}(T-t,x_t\mid 0,x_T)$ is an approximation of the marginal density $p^{\star}(t,x_t)$ in \eqref{eqn:marginal_bridge} for each $t\in(0,T)$.  

Our discussion in \eqref{eqn:recover_htransform} prompts having the time-reversal of \eqref{eqn:approx_reverse} as an approximation of the Doob's $h$-transform process $X^{\star}$ in \eqref{eqn:doob_SDE}. The bridge process $\hat{X}=(\hat{X}_t)_{t\in[0,T]}=(\hat{Z}_{T-t})_{t\in[0,T]}$ satisfies 
\begin{align*}
	\rmd \hat{X}_t = \hat{f}(t,\hat{X}_t)\rmd t + \sigma(t,\hat{X}_t)\rmd W_t,\quad \hat{X}_0 = x_0,
\end{align*}
with drift function $\hat{f}(t,x_t) = -\hat{b}(T-t,x_t) + \Sigma(t,x_t)\hat{s}^{\star}(t,x_t) + \nabla\cdot\Sigma(t,x_t)$ which
is to be seen as an approximation of $f^{\star}(t,x_t)$. 
We can approximate the score $\hat{s}^{\star}(t,x_t)=\nabla\log \hat{p}^{\star}(t,x_t)$ of the marginal density $\hat{p}^{\star}(t,x_t)$ and hence the time-reversal of \eqref{eqn:approx_reverse} using the methodology described in Section \ref{sec:learn_reversal}. The following summarizes the key elements involved.

Consider a path measure $\hat{\mathbb{Q}}_{\phi}^{x_0,x_T}$ that is induced by the bridge process $\hat{X}=(\hat{X}_t)_{t\in[0,T]}$ satisfying 
\begin{align}\label{eqn:proposal_bridge_process}
	\rmd \hat{X}_t = \hat{f}_{\phi}(t,\hat{X}_t)\rmd t + \sigma(t,\hat{X}_t)\rmd W_t,\quad \hat{X}_0=x_0,
\end{align}
with drift function $\hat{f}_{\phi}(t,x_t) = -\hat{b}(T-t,x_t) + \Sigma(t,x_t)\hat{s}^{\star}_{\phi}(t,x_t) + \nabla\cdot\Sigma(t,x_t)$, where $\hat{s}_{\phi}^{\star}:[0,T]\times\mathbb{R}^d\rightarrow\mathbb{R}^d$ denotes a function approximation of the score $\hat{s}^{\star}(t,x_t)$ with parameters $\phi\in\Phi$ to be optimized. We now measure the score approximation error as 
\begin{align*}
    \hat{e}(\phi) = \hat{E}^{x_0,x_T}\left\lbrace\int_0^T\left\|\hat{s}_{\phi}^{\star}(t,\hat{X}_t)-\hat{s}^{\star}(t,\hat{X}_t)\right\|_{\Sigma(t,\hat{X}_t)}^2\rmd t\right\rbrace.
\end{align*}

\begin{proposition}
\label{prop:forward_KL}
Assuming $\hat{e}(\phi)<\infty$, we have  $\textsc{kl}(\hat{\mathbb{Q}}^{x_0,x_T}|\hat{\mathbb{Q}}_{\phi}^{x_0,x_T}) \leq \hat{e}(\phi) / 2$ and $\hat{X}_T=x_T$ holds $\hat{\mathbb{Q}}_{\phi}^{x_0,x_T}$-almost surely. 
\end{proposition}

\begin{proposition}\label{prop:denoise_hat}
For any partition $(t_m)_{m=0}^M$ of the interval $[0,T]$, 
we have $\textsc{kl}(\hat{\mathbb{Q}}^{x_0,x_T}|\hat{\mathbb{Q}}_{\phi}^{x_0,x_T}) \leq \hat{L}(\phi)+\hat{C}$ if $\hat{e}(\phi)<\infty$, where $\hat{C}$ is a constant independent of $\phi\in\Phi$, the loss function is defined as 
\begin{align}\label{eqn:loss_hat}
	\hat{L}(\phi)= \frac{1}{2}\sum_{m=1}^M\int_{t_{m-1}}^{t_m} \hat{E}^{x_0,x_T}\left\lbrace\left\|\hat{s}_{\phi}^{\star}(T-t,\hat{Z}_t) - \hat{g}(t_{m-1},\hat{Z}_{t_{m-1}},t,\hat{Z}_t) \right\|_{\Sigma(t,\hat{Z}_t)}^2\right\rbrace\rmd t,
\end{align}
and $\hat{g}(s,z_s,t,z_t)=\nabla\log \hat{q}(t,z_t\mid s,z_s)$ for $0\leq s<t\leq T$ and $z_s,z_t\in\mathbb{R}^d$. 
\end{proposition}

The proof of these results is similar to Section \ref{sec:proofs_backward} of the Supplementary Material and is thus omitted. As before, this allows us to circumvent intractability in the Kullback--Leibler divergence $\textsc{kl}(\hat{\mathbb{Q}}^{x_0,x_T}|\hat{\mathbb{Q}}_{\phi}^{x_0,x_T})$ by minimizing the loss function $\hat{L}(\phi)$. 
In the ideal case of $\hat{s}_{\phi}^{\star}(t,x_t)=\hat{s}^{\star}(t,x_t)$ $\hat{\mathbb{Q}}^{x_0,x_T}$-almost surely, the minimal loss of $\hat{L}(\phi)=-\hat{C}$ is also unknown in practice, and $\hat{\mathbb{Q}}_{\phi}^{x_0,x_T}$ recovers the law of the diffusion bridge process $X^{\star}$ only if the initial score approximation error satisfies $e(\hat{\phi})=0$. 

Given a minimizer $\phi^{\circ}\in\arg\min_{\phi\in\Phi}\hat{L}(\phi)$ and the corresponding score approximation $s^{\circ}(t,x_t)=\hat{s}_{\phi^{\circ}}^{\star}(t,x_t)$, by rewriting the drift $f^{\circ}(t,x_t)=\hat{f}_{\phi^{\circ}}(t,x_t)$ as 
\begin{align*}
	f^{\circ}(t,x_t) = f(t,x_t) + \Sigma(t,x_t)\{s^{\circ}(t,x_t) - \hat{s}(t,x_t)\}
\end{align*}
and comparing it with \eqref{eqn:recover_htransform}, we see that the last two terms on the right provide an approximation of the term $\Sigma(t,x_t)\nabla\log h(t,x_t)$ in Doob's $h$-transform.

\section{Numerical implementation}\label{sec:numerical_implement}
In this section, we detail various numerical considerations to implement our proposed methodology. 
For simplicity, we employ the Euler--Maruyama scheme \citep[p. 340]{kloeden1992} on a uniform discretization of the interval $[0,T]$, denoted by $0=t_0<t_1<\cdots<t_M=T$, with stepsize $\delta t = T/M$ for $M>1$. 
Non-uniform discretizations involve only minor modifications; some higher-order schemes could also be considered. 
In the following, we denote a multivariate normal distribution with mean vector $\mu\in\mathbb{R}^d$ and covariance matrix $\Sigma\in\mathbb{R}^{d\times d}$ as $\mathcal{N}(\mu,\Sigma)$, and its density as $x\mapsto\mathcal{N}(x;\mu,\Sigma)$. We write the zero vector as $0_d\in\mathbb{R}^d$ and the identity matrix as $I_d\in\mathbb{R}^{d\times d}$. 

The time-discretization of the stochastic differential equation defining $X$ satisfies the following recursion
\begin{align}\label{eqn:euler}
    X_{t_m} = X_{t_{m-1}} + \delta t f(t_{m-1},X_{t_{m-1}}) + 
    \sigma(t_{m-1},X_{t_{m-1}}) (W_{t_m} - W_{t_{m-1}}),\quad X_0=x_0,
\end{align}
for $m\in\{1,\ldots,M\}$, with independent Brownian increments $W_{t_m} - W_{t_{m-1}}\sim\mathcal{N}(0_d,\delta t I_d)$.
Equation \eqref{eqn:euler} induces a normal approximation of the transition density $p(t_m,x_{t_m}\mid t_{m-1},x_{t_{m-1}})$ of the form
\begin{align*}
    p_M(t_m,x_{t_m}\mid t_{m-1},x_{t_{m-1}})=\mathcal{N}\{x_{t_m}; x_{t_{m-1}} + \delta t f(t_{m-1},x_{t_{m-1}}), \delta t\Sigma(t_{m-1},x_{t_{m-1}})\}.
\end{align*}
By replacing the score of $p(t_m,x_{t_m}\mid t_{m-1},x_{t_{m-1}})$ with that of $p_M(t_m,x_{t_m}\mid t_{m-1},x_{t_{m-1}})$, $g(t_{m-1},x_{t_{m-1}},t_m,x_{t_m})$ in Proposition \ref{prop:denoise} can be approximated by 
\begin{align*}
    g_M(t_{m-1},x_{t_{m-1}},t_m,x_{t_m}) 
    = -(\delta t)^{-1}\Sigma^{-1}(t_{m-1},x_{t_{m-1}})\{x_{t_m} - x_{t_{m-1}} - \delta t f(t_{m-1},x_{t_{m-1}})\}.\notag
\end{align*}
This approximation is of order $(\delta t)^{-1/2}$ as $\delta t\rightarrow 0$. 
One could consider a control variate approach to stabilize the approximation as discussed in \cite{song2021train}. While one could also consider higher-order discretization schemes to approximate the function $g$, it may not be feasible or worthwhile to compute higher-order derivatives of $f$ and $\sigma$, particularly if the time-discretization error is dominated by the neural network approximation error.  
 
We then define the following approximation of the loss function $L(\phi)$ in \eqref{eqn:loss} 
\begin{equation*}
	L_M(\phi) = \frac{1}{2}\delta t\sum_{m=1}^M {E}_M^{x_0}\left\lbrace\left\|s_{\phi}(t_m,X_{t_m}) - g_M(t_{m-1},X_{t_{m-1}},t_m,X_{t_m}) \right\|_{\Sigma(t_m,X_{t_m})}^2\right\rbrace,
\end{equation*}
where ${E}_M^{x_0}$ denotes expectation with respect to the law of the time-discretized process under \eqref{eqn:euler}. 
To obtain a minimizer $\hat{\phi}\in\arg\min_{\phi\in\Phi}L_M(\phi)$ using stochastic gradient algorithms, the gradient with respect to parameters $\phi\in\Phi$ 
\begin{align}\label{eqn:gradient_parameters_loss}
    &\nabla_{\phi}L_M(\phi) = \delta t\sum_{m=1}^M E_M^{x_0}\left[(\nabla_{\phi}s_{\phi}^\T \Sigma)(t_m,X_{t_m})
    \{s_{\phi}(t_m,X_{t_m}) - g_M(t_{m-1},X_{t_{m-1}},t_m,X_{t_m})\}\right]
\end{align}
can be unbiasedly estimated using independent sample paths from \eqref{eqn:euler}. 
The above notation $\nabla_{\phi}s_{\phi}$ refers to the Jacobian of $s_{\phi}$. 
Equation \eqref{eqn:gradient_parameters_loss} can be seen as an approximation of the gradient $\nabla_{\phi}{L}(\phi)$.

After obtaining $\hat{\phi}$ with optimization, we can simulate a proposal bridge $\hat{Z}=(\hat{Z}_t)_{t\in[0,T]}$ satisfying \eqref{eqn:approx_reverse} with drift function $\hat{b}(t,z_t)=b_{\hat{\phi}}(t,z_t)$. We employ the following modified Euler--Maruyama scheme 
\begin{align}\label{eqn:discretized_backward_process}
    \hat{Z}_{t_m} = \hat{Z}_{t_{m-1}} + \delta t \hat{b}(t_{m-1},\hat{Z}_{t_{m-1}}) + \left(\frac{T-t_m}{T-t_{m-1}}\right)^{1/2}\sigma(T-t_{m-1},\hat{Z}_{t_{m-1}}) (B_{t_m} - B_{t_{m-1}}),
\end{align}
for $m\in\{1,\ldots,M-1\}$, with independent Brownian increments $B_{t_m} - B_{t_{m-1}}\sim\mathcal{N}(0_d,\delta t I_d)$, initial condition $\hat{Z}_0=x_T$, and terminal constraint $\hat{Z}_T=x_0$. This changes the variance of the usual Euler--Maruyama transitions with a multiplier of $(T-t_m)/(T-t_{m-1})$ at time step $m$. We found that this modification can improve practical performance for times near the endpoint by lowering the transition variances. Such behaviour is consistent with findings in earlier works by \cite{durham2002numerical} and \cite{papaspiliopoulos2013data} when constructing proposal bridge processes with the drift of a Brownian bridge. This gives a normal approximation of the transition density $\hat{q}(t_m,z_{t_m}\mid t_{m-1},z_{t_{m-1}})$ 
\begin{align}\label{eqn:normal_backward_transition}
	&\hat{q}_M(t_m,z_{t_m}\mid t_{m-1},z_{t_{m-1}})\notag\\
	&=\mathcal{N}\left\lbrace z_{t_m}; z_{t_{m-1}} + \delta t \hat{b}(t_{m-1},z_{t_{m-1}}), \delta t\left(\frac{T-t_m}{T-t_{m-1}}\right)\Sigma(T-t_{m-1},z_{t_{m-1}})\right\rbrace.
\end{align}
We can perform importance sampling on $\mathbb{E}_M=(\mathbb{R}^d)^{M-1}$ to correct for the discrepancy between the law of our proposal bridge process 
\begin{align}\label{eqn:law_proposal_bridge}
	\hat{Q}_M^{x_0,x_T}(z_{t_1:t_{M-1}}) = \prod_{m=1}^{M-1}\hat{q}_M(t_m,z_{t_m}\mid t_{m-1},z_{t_{m-1}}),\quad z_{t_1:t_{M-1}}=(z_{t_m})_{m=1}^{M-1}\in\mathbb{E}_M,
\end{align}
and the law of the time-discretized diffusion bridge process 
\begin{align}\label{eqn:law_diffusion_bridge}
	P_M^{x_0,x_T}(x_{t_1:t_{M-1}}) = \frac{\gamma_M^{x_0,x_T}(x_{t_1:t_{M-1}})}{p_M(T,x_T\mid 0,x_0)},
	\quad x_{t_1:t_{M-1}}=(x_{t_m})_{m=1}^{M-1}\in\mathbb{E}_M,
\end{align}
with $\gamma_M^{x_0,x_T}(x_{t_1:t_{M-1}})=\prod_{m=1}^M p_M(t_m,x_{t_m}\mid t_{m-1},x_{t_{m-1}})$, and also estimate 
\begin{align}\label{eqn:discretized_transition}
	p_M(T,x_T\mid 0,x_0) = \int_{\mathbb{E}_M} \gamma_M^{x_0,x_T}(x_{t_1:t_{M-1}}) \rmd x_{t_1:t_{M-1}},	
\end{align}
which is an approximation of the transition density $p(T,x_T\mid 0,x_0)$ under the Euler--Maruyama scheme. 
The corresponding unnormalized importance weight is $\omega(z_{t_1:t_{M-1}})=\gamma_M^{x_0,x_T}(x_{t_1:t_{M-1}})/\hat{Q}_M^{x_0,x_T}(z_{t_1:t_{M-1}})$ with $x_{t_1:t_{M-1}}=(z_{T-t_m})_{m=1}^{M-1}$, and an unbiased importance sampling estimator of the transition density $p_M(T,x_T\mid 0,x_0)$ is $N^{-1}\sum_{n=1}^N\omega(z_{t_1:t_{M-1}}^n)$ where $(z_{t_1:t_{M-1}}^n)_{n=1}^N$ denote $N\in\mathbb{N}$ independent sample paths from $\hat{Q}_M^{x_0,x_T}$. 
As noted by \cite{lin2010generating}, the root mean squared error of this transition density estimator is approximately equals to the $\chi^2$-divergence of $\hat{Q}_M^{x_0,x_T}$ from $P_M^{x_0,x_T}$ divided by the sample size $N$.
One can also employ proposals from \eqref{eqn:law_proposal_bridge} within an independent Metropolis--Hastings algorithm that has \eqref{eqn:law_diffusion_bridge} as its invariant law  \citep{elerian2001likelihood}. At each iteration of the algorithm, a sample path $z_{t_1:t_{M-1}}^{\circ}\sim\hat{Q}_M^{x_0,x_T}$ is accepted with probability $\min\{1,\omega(z_{t_1:t_{M-1}}^{\circ})/\omega(z_{t_1:t_{M-1}})\}$, where $z_{t_1:t_{M-1}}$ denotes the current state of the Markov chain. The efficiency of this Markov chain Monte Carlo algorithm can be assessed by monitoring its acceptance probability.
To improve the acceptance probability, we can also combine independent Metropolis--Hastings with importance sampling within a particle independent Metropolis--Hastings algorithm \citep{andrieu2010particle} that has invariant law \eqref{eqn:law_diffusion_bridge}. Each iteration of this algorithm involves selecting a proposed sample path $z_{t_1:t_{M-1}}^{\circ}$ among $N$ candidates $(z_{t_1:t_{M-1}}^n)_{n=1}^N\sim\hat{Q}_M^{x_0,x_T}$ according to probabilities proportional to their weights $(\omega(z_{t_1:t_{M-1}}^n))_{n=1}^N$, and accepting it with probability $\min\{1,\hat{p}_M^{\circ}(T,x_T\mid 0,x_0)/\hat{p}_M(T,x_T\mid 0,x_0)\}$ that depends on the ratio of the new and current transition density estimators $\hat{p}_M^{\circ}(T,x_T\mid 0,x_0)=N^{-1}\sum_{n=1}^N\omega(z_{t_1:t_{M-1}}^n)$ and 
$\hat{p}_M(T,x_T\mid 0,x_0)$, respectively. Under mild assumptions, consistency of importance sampling estimators as $N\rightarrow\infty$ implies that the acceptance probability of particle independent Metropolis--Hastings algorithm converges to one. This algorithm can also be combined with unbiased Markov chain Monte Carlo methods to provide unbiased estimates of expectations with respect to the law of the time-discretized diffusion bridge \citep{middleton2019}. 

Lastly, we sketch the key steps to learn the Doob's $h$-transform process for the sake of brevity. Using the score of the normal transition density in \eqref{eqn:normal_backward_transition}, we may approximate $\hat{g}(t_{m-1},z_{t_{m-1}},t_m,z_{t_m})$ and hence the loss function $\hat{L}(\phi)$ in \eqref{eqn:loss_hat}. 
The approximate loss can be minimized using stochastic gradient algorithms and sample paths from \eqref{eqn:discretized_backward_process}. By time-discretizing the resulting proposal bridge process in \eqref{eqn:proposal_bridge_process}, we may then employ it as an importance proposal to approximate the law in \eqref{eqn:law_diffusion_bridge} and the transition density in \eqref{eqn:discretized_transition}, or to generate proposal distributions within independent Metropolis--Hastings algorithms.   

\section{Implementation details}\label{sec:implement_details}
\subsection{Benchmarking proposal bridge processes}
We benchmark our diffusion bridge approximations against several existing approaches to construct proposal bridge processes. These methods simulate a proposal bridge process $X^{\circ}=(X_t^{\circ})_{t\in[0,T]}$ satisfying 
\begin{align}\label{eqn:propsal_bridge}
    \rmd X_t^{\circ} = f^{\circ}(t,X_t^{\circ}) \rmd t + \sigma(t,X_t^{\circ})\rmd W_t, 
    \quad X_0^{\circ}=x_0.
\end{align}
As summarized in Table \ref{tab:benchmark}, each method can be understood as a specific choice of the drift function $f^{\circ}$ that approximates the diffusion bridge process $X^{\star}$ given by Doob's $h$-transform in \eqref{eqn:doob_SDE}. 
We time-discretize \eqref{eqn:propsal_bridge} using the Euler--Maruyama (EM) scheme \eqref{eqn:euler} for the forward diffusion method, the modified Euler--Maruyama (Modified EM) scheme \eqref{eqn:discretized_backward_process} for the modified diffusion bridge, the Euler--Maruyama scheme 
with the time-change $\tau(t)=t(2-t/T)$ proposed by \citet{van2017bayesian} (Time-change EM) for the Clark--Delyon--Hu bridge and the guided proposal bridge.

We perform an importance sampling or independent Metropolis--Hastings correction as described above, with the exception that for the Clark--Delyon--Hu bridge and the guided proposal bridge, the importance weight 
\begin{align*}
    \omega(X_{\tau(t_1):\tau(t_{M-1})}^\circ) = \frac{\tilde{p}(T,x_T\mid 0,x_0)}{p(T,x_T\mid 0,x_0)}\exp\left[\sum_{m=1}^M (f-\tilde{f})^\T(\nabla\log \tilde{h})\{\tau(t_m), X_{\tau(t_m)}^{\circ}\}\right]  
\end{align*}
of the sample path $(X_{\tau(t_m)}^\circ)_{m=1}^{M-1}$ is obtained by approximating the Radon--Nikodym derivative 
\begin{align}\label{eqn:radon_guiding}
    \Psi(X^\circ)= \frac{\tilde{p}(T,x_T\mid 0,x_0)}{p(T,x_T\mid 0,x_0)}\exp\left[\int_0^T (f - \tilde{f})^\T(\nabla\log \tilde{h})\{t, X_t^{\circ}\}\rmd t\right],
\end{align}
where $\tilde{f}(t,x_t)$ denotes the drift function of the associated auxiliary process. The numerical results in \citet[Section 5.2]{van2017bayesian} show improved time-discretization of $\Psi(X^\circ)$ using Euler--Maruyama with the time-change $\tau(t)=t(2-t/T)$. 
The Clark--Delyon--Hu bridge can be understood as having Brownian motion as the auxiliary process, in which case $\tilde{f}(t,x_t)=0$, while the guided proposal bridge typically involves selecting an auxiliary Ornstein--Uhlenbeck process with a linear drift $\tilde{f}(t,x_t)$ whose parameters are determined by minimizing a Kullback--Leibler objective \citep[Section 1.3]{schauer2017guided} or by understanding the behaviour of the diffusion process $X$ \citep[Section 4.4]{van2017bayesian}. We will detail the choice of these parameters for each example in the following.

\begin{table}
\footnotesize
\begin{tabular}{lcccc}
\\
\textbf{Method} & \textbf{Drift} $f^{\circ}(t,x_t)$ & \textbf{References} & \textbf{Time-discretization} \\ \hline 
Forward diffusion & $f(t,x_t)$ & \citet{pedersen1995consistency} & EM \\
Modified diffusion bridge & $(x_T-x_0)/(T-t)$ & \citet{durham2002numerical} & Modified EM \\
Clark--Delyon--Hu & $f(t,x_t)+(x_T-x_0)/(T-t)$ & \citet{clark1990simulation,delyon2006simulation} & Time-change EM \\
Guided proposal & $f(t,x_{t})+\Sigma(t,x_{t})\nabla\log\tilde{h}(t,x_{t})$ & \citet{schauer2017guided} & Time-change EM
\end{tabular}
\caption{Benchmark proposal bridge methods\label{tab:benchmark}}
\end{table}

\subsection{Neural network and stochastic optimization}\label{sec:nn}
The architecture of the neural networks we employed is illustrated in Fig. \ref{fig:net}. For all numerical experiments, optimization was performed using the stochastic gradient algorithm of \citet{kingma2014adam} with a momentum of $0.99$ and learning rate of $0.01$. 

\begin{figure}
\centering
\includegraphics[width=1\linewidth]{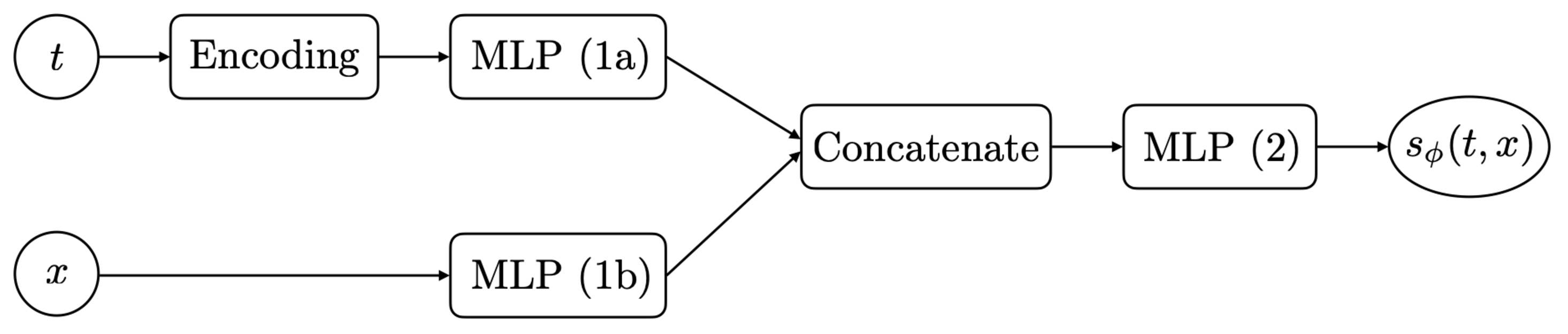}
\includegraphics[width=1\linewidth]{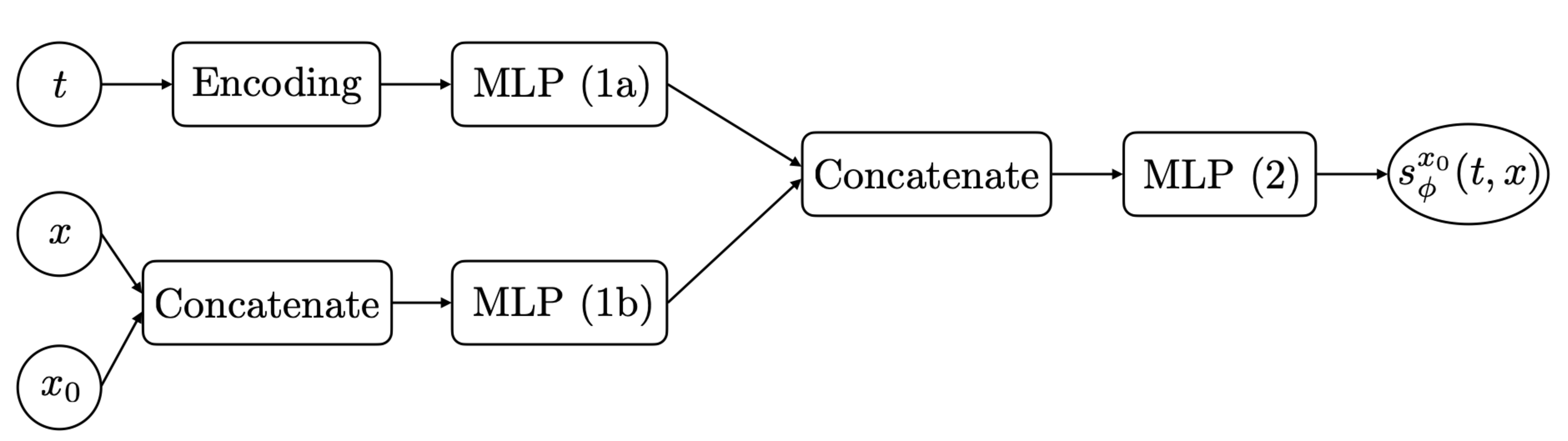}
\caption{Neural network architectures involve multi-layer perceptron (MLP) blocks and an ``Encoding'' block which applies the sine transform described in \citet{vaswani2017attention}. MLP (1a) and (1b) have one hidden layer and MLP (2) has two hidden layers. All neurons use the Leaky ReLU activation function.}
\label{fig:net}
\end{figure}

\subsection{Ornstein--Uhlenbeck process}
Let $X$ be an Ornstein--Uhlenbeck process, defined by \eqref{eqn:SDE} with linear drift function $f(t,x_t)= \alpha-\beta x_t$ and identity diffusion coefficient $\sigma(t,x_t) = I_d$. 
In this analytically tractable example, for any $0\leq s<t\leq T$, the transition density of $X$ is a normal density $p(t,x_t\mid s,x_s)=\mathcal{N}\{x_t;m(t-s,x_s),v(t-s)I_d\}$, with the following mean and variance 
\begin{align*}
	m(t-s,x_s) = \frac{\alpha}{\beta} + \left(x_s - \frac{\alpha}{\beta}\right)\exp\{-\beta (t-s)\},\quad v(t-s) =\frac{1-\exp\{-2\beta (t-s)\}}{2\beta}.
\end{align*}
Hence the logarithmic gradient term in the Doob's $h$-transform of \eqref{eqn:doob_SDE} is 
\begin{align*}
    \nabla \log h(t,x_t) = \frac{\exp\{-\beta(T-t)\}}{v(T-t)}\{x_T - m(T-t,x_t)\},
\end{align*} 
and the score of the transition density in \eqref{eqn:reversal_bridge} is
\begin{align*}
	s(t,x_t) = v(t)^{-1}\{m(t,x_0)-x_t\}.
\end{align*}

For this example, we can select the auxiliary process of \citet{schauer2017guided} as the Ornstein--Uhlenbeck process by setting $\tilde{f}(t,x_t)=\alpha-\beta x_t$, in which case the linear guiding term $\nabla \log \tilde{h}(t,x_t)=\nabla \log h(t,x_t)$ is exact, and the Radon--Nikodym derivative in \eqref{eqn:radon_guiding} satisfies $\Psi(X^\circ)=1$. 

For dimension $d=1$ and varying either the time horizon $T\in\{1,2,4,8\}$ or the terminal state $x_T$, we employed a time-discretization stepsize of $\delta t = 0.02$, $500$ optimization iterations, and $100$ sample paths per iteration. For the case of $T=1$ and varying $d\in\{1,2,4,8\}$, we decreased the stepsize and increased the number of optimization iterations and the capacity of the neural network with dimension.

\subsection{Interest rates model}
We consider a special case of an interest rates model in \citet{ait1998nonparametric}, defined by \eqref{eqn:SDE} with drift function $f(t,x_t)=\theta/x_t-x_t$ with $\theta=4$ and diffusion coefficient $\sigma(t,x_t)=1$. This diffusion has a stable stationary point at $x^{\star}=\sqrt{\theta}$, and its transition density is known and given by 
\begin{align*}
    \log p(t,x_{t}\mid s,x_{s})
	&=\theta\log(x_{t}/x_{s})+\frac{1}{2}\log(x_{t}x_{s})-x_{t}^{2}+\left(\theta+\frac{1}{2}\right)(t-s)\notag\\
	&-\log\mathrm{sinh}(t-s)
	-\frac{x_{t}^{2}+x_{s}^{2}}{\exp\{2(t-s)\}-1}+\log I_{\theta-1/2}\left\lbrace\frac{x_{t}x_{s}}{\mathrm{sinh}(t-s)}\right\rbrace,
\end{align*}
for $0\leq s<t\leq T$, where $I_{\nu}$ denotes the modified Bessel function of order $\nu$. 
The logarithmic gradient term in the Doob's $h$-transform of \eqref{eqn:doob_SDE} is 
\begin{align*}
    \nabla \log h(t,x_t) 
	&=-\frac{\theta}{x_{t}}+\frac{1}{2x_{t}}-\frac{2x_{t}}{\exp\{2(T-t)\}-1}\notag\\
	&+I_{\theta-1/2}\left\lbrace\frac{x_{T}x_{t}}{\mathrm{sinh}(T-t)}\right\rbrace^{-1}J_{\theta-1/2}\left\lbrace\frac{x_{T}x_{t}}{\mathrm{sinh}(T-t)}\right\rbrace\frac{x_{T}}{\mathrm{sinh}(T-t)}, 
\end{align*}
and the score of the transition density in \eqref{eqn:reversal_bridge} is
\begin{align*}
    s(t,x_t)
	=\frac{\theta}{x_{t}}+\frac{1}{2x_{t}}-2x_{t}-\frac{2x_{t}}{\exp(2t)-1}+I_{\theta-1/2}\left\lbrace\frac{x_{t}x_{0}}{\mathrm{sinh}(t)}\right\rbrace^{-1}J_{\theta-1/2}\left\lbrace\frac{x_{t}x_{0}}{\mathrm{sinh}(t)}\right\rbrace\frac{x_{0}}{\mathrm{sinh}(t)},
\end{align*}
where $J_{\nu}$ denotes the derivative of $I_{\nu}$. 

Numerical experiments for all $T\in\{1,2,4,8\}$ employed a time-discretization stepsize of $\delta t=0.02$, $1000$ optimization iterations, $1000$ sample paths
per iteration with $10$ unique initial conditions $X_0=x_0$ sampled from the gamma distribution with shape $5$ and rate $2$. 
We select the auxiliary process of \citet{schauer2017guided} as an Ornstein--Uhlenbeck process with drift $\tilde{f}(t,x_t)=2\sqrt{\theta}-2 x_t$ and unit diffusion coefficient. The choice of $\tilde{f}$ is based on the first-order Taylor approximation
\begin{align*}
f(t,x_t)\approx f(t, x^{\star}) + \partial_x f(t, x^{\star}) (x_t - x^{\star}) = \tilde{f}(t,x_t).
\end{align*}

\subsection{Cell model}
Our numerical experiments for all $T\in\{2,4,8,16\}$ and $\sigma_X^2\in\{0.1,1\}$ employed a time-discretization stepsize of $\delta t=0.02$, $2000$ optimization iterations, and $100$ sample paths per iteration. 
In our implementation of the guided proposal bridge, we choose an auxiliary Ornstein--Uhlenbeck process with drift function $\tilde{f}(t,x_t)=A(\theta-x_t)$ with $\theta\in\mathbb{R}^d$ and $A=UDU^\T\in\mathbb{R}^{d\times d}$ parameterized by a matrix $U\in\mathbb{R}^{d\times d}$ whose columns are eigenvectors and a diagonal matrix $D\in\mathbb{R}^{d\times d}$ whose diagonal entries are eigenvalues. This eigendecomposition facilities computation of matrix exponentials appearing in the expressions of the transition density $\tilde{p}(t,x_t\mid 0,x_0)$ and the gradient $\nabla\log \tilde{h}(t,x_t)$. 

Following \citet{schauer2017guided}, we optimize the parameters $\phi=(\theta, U, D)$ by minimizing the Kullback--Leibler divergence 
$\textsc{kl}(\mathbb{P}^{x_0,x_T}|\mathbb{Q}_{\phi}^{x_0,x_T})$, where $\mathbb{P}^{x_0,x_T}$ denotes the law of the diffusion bridge and $\mathbb{Q}_{\phi}^{x_0,x_T}$ is the law induced by the guided proposal bridge process in \eqref{eqn:propsal_bridge}. 
As considered in \citet{schauer2017guided}, we employ importance sampling to approximate the intractable objective by rewriting it as 
\begin{align}\label{eqn:guided_KL_objective}
    \textsc{kl}(\mathbb{P}^{x_0,x_T}|\mathbb{Q}_{\phi}^{x_0,x_T}) = 
    E_{\varphi}^{x_0,x_T}\{\log (\rmd \mathbb{P}^{x_0,x_T}/\rmd \mathbb{Q}_{\phi}^{x_0,x_T})(\rmd \mathbb{P}^{x_0,x_T}/\rmd \mathbb{Q}_{\varphi}^{x_0,x_T})(X^\circ)\},
\end{align}
where $E_{\varphi}^{x_0,x_T}$ denotes expectation with respect to the proposal law $\mathbb{Q}_{\varphi}^{x_0,x_T}$ with reference parameters $\varphi$ that are obtained from earlier iterations of a stochastic gradient descent algorithm. 

Initialization of $\varphi$ is crucial as estimators of the gradient of \eqref{eqn:guided_KL_objective} with respect to $\phi$ will have large variance when the importance sampling approximation under $\mathbb{Q}_{\varphi}^{x_0,x_T}$ is poor. For this application, we initialize by setting $(\theta, U, D)=(X_T, I_d, I_d)$, corresponding to starting $A$ with the identity matrix and choosing $\theta$ to induce mean-reversion towards the stable stationary point $X_T$. 
After each gradient update, we also perform projection to ensure that columns of $U$ are orthogonal.

\end{document}